\documentclass[pra,aps,superscriptaddress,twocolumn,letter,nopacs,nofootinbib,longbibliography]{revtex4-2}

\usepackage{array,multirow,graphicx,color,amsmath,amsfonts,enumerate,amsthm,amssymb,mathtools,enumitem,thmtools,hyperref,subfigure,mathdots,enumitem,centernot,bm,soul,bbm,tikz,pgfplots}
\usepackage[normalem]{ulem}
\usepackage[capitalise, noabbrev]{cleveref}
\usetikzlibrary{arrows}
\pgfplotsset{compat=1.14}
\hypersetup{colorlinks=true,linkcolor=blue,citecolor=blue,urlcolor=blue}

\def\D{ {\cal D} }
\def\E{ {\cal E} }

\def\H{ {\cal H} }
\def\I{ {\cal I} }
\def\L{ {\cal L} }

\def\>{\rangle}
\def\<{\langle}
\newcommand{\bra}[1]{\langle {#1} |}
\newcommand{\ket}[1]{| {#1} \rangle}
 
\newcommand{\ketbra}[2]{\ensuremath{\left|#1\right\rangle\!\!\left\langle#2\right|}}

\newcommand{\matrixel}[3]{\ensuremath{\left\langle #1 \vphantom{#2#3} \right| #2 \left| #3 \vphantom{#1#2} \right\rangle}}
\newcommand{\tr}[1]{\mathrm{Tr}\left( #1 \right)}

\newcommand{\iden}{\mathbbm{1}}

\renewcommand{\v}[1]{\ensuremath{\boldsymbol #1}}

\definecolor{ppblue}{RGB}{46,117,182}
\definecolor{ppred}{RGB}{197, 90, 17}

\theoremstyle{plain}
\newtheorem{thm}{Theorem}
\newtheorem{res}{Result}
\newtheorem{lem}[thm]{Lemma}
\newtheorem{ex}{Example}

\newtheorem{cor}[thm]{Corollary}

\theoremstyle{definition}
\newtheorem{defn}{Definition}

\newcolumntype{C}[1]{>{\centering\arraybackslash}p{#1}}

\definecolor{tikzBlue}{rgb}{0.6941176470588235,0.7568627450980392,0.8588235294117647}
\definecolor{tikzOrange}{rgb}{0.9294117647058824,0.7647058823529411,0.49019607843137253}
\definecolor{tikzBlue2}{rgb}{0.462745098,0.504575163,0.57254902}
\definecolor{tikzOrange2}{rgb}{0.619607843,0.509803922,0.326797386}
\definecolor{tikzGray}{rgb}{0.7529411764705882,0.7529411764705882,0.7529411764705882}

\DeclareFontFamily{U}{mathb}{\hyphenchar\font45}
\DeclareFontShape{U}{mathb}{m}{n}{
	<-6> mathb5 <6-7> mathb6 <7-8> mathb7
	<8-9> mathb8 <9-10> mathb9
	<10-12> mathb10 <12-> mathb12
}{}
\DeclareSymbolFont{mathb}{U}{mathb}{m}{n}
\DeclareMathSymbol{\llcurly}{\mathrel}{mathb}{"CE}
\DeclareMathSymbol{\ggcurly}{\mathrel}{mathb}{"CF}

\begin{document}

\title{Quantum advantage in simulating stochastic processes}
\author{Kamil Korzekwa$^{*,}$}
\affiliation{Faculty of Physics, Astronomy and Applied Computer Science, Jagiellonian University, 30-348 Krak\'{o}w, Poland}
\affiliation{International Centre for Theory of Quantum Technologies, University of Gda{\'n}sk, 80-308 Gda{\'n}sk, Poland}
\author{Matteo Lostaglio$^{*,}$}
\affiliation{ICFO-Institut de Ciencies Fotoniques, The Barcelona Institute of Science and Technology, Castelldefels (Barcelona), 08860, Spain}
\affiliation{QuTech, Delft University of Technology, P.O. Box 5046, 2600 GA Delft, Netherlands}
\footnotetext{These authors contributed equally to this work. \\ Emails: korzekwa.kamil@gmail.com, lostaglio@protonmail.com}

\begin{abstract}
	We investigate the problem of simulating classical stochastic processes through quantum dynamics, and present three scenarios where memory or time quantum advantages arise. First, by introducing and analysing a quantum version of the embeddability problem for stochastic matrices, we show that quantum memoryless dynamics can simulate classical processes that necessarily require memory. Second, by extending the notion of space-time cost of a stochastic process $P$ to the quantum domain, we prove an advantage of the quantum cost of simulating $P$ over the classical cost. Third, we demonstrate that the set of classical states accessible via Markovian master equations with quantum controls is larger than the set of those accessible with classical controls, leading, e.g., to a potential advantage in cooling protocols. 
\end{abstract}

\maketitle

%------------------------------------------------------------------

\section{Introduction}

\subsection{Memory advantages}

What tasks can we perform more efficiently by employing quantum properties of nature? And what are the quantum resources powering them? These are the central questions that need to be answered not only to develop novel quantum technologies, but also to deepen our understanding of the foundations of physics. Over the last few decades, these questions were successfully examined in the context of cryptography~\cite{gisin2002quantum}, computing~\cite{nielsen2010quantum}, simulations~\cite{georgescu2014quantum} and sensing~\cite{degen2017quantum}, proving that the quantum features of nature can indeed be harnessed to our benefit. 

More recently, an area of active theoretical and experimental interest focused on the memory advantages offered by quantum mechanics for the simulation of stochastic processes in the setting of classical causal models~\cite{suen2017classical, thompson2018causal,binder2018practical, ghafari2019dimensional}. An experimentally accessible and relevant measure of such an advantage is the dimensionality of the memory required for the simulation~\cite{ghafari2019dimensional, elliott2020extreme}. These dimensional advantages have been identified experimentally (a qubit system has been used to simulate a stochastic process that classically requires three bits~\cite{ghafari2019dimensional}), and theoretically for a certain class of Poisson processes~\cite{elliott2020extreme}.

Here we take a complementary approach starting from the following simple observation: although all fundamental interactions are memoryless, the basic information-processing primitives (such as the bit-swap operation) cannot be performed classically in a time-continuous fashion without employing implicit microscopic states that act as a memory~\cite{wolpert2019space}. We show that this picture changes dramatically if instead we consider memoryless \emph{quantum} dynamics. This is due to quantum coherence, arising from the superposition principle, which can effectively act as an internal memory of the system during the evolution. 

%------------------------------------------------------------------

\subsection{Classical vs quantum}
\label{sec:intro:b}

But what do we really mean when we say that a bit-swap (or other information processing tasks) cannot be performed classically in a memoryless way? First, when we speak of a bit we mean a fundamentally two-level system, i.e., a system with only two microscopic degrees of freedom (e.g. a spin-1/2 particle) and not a macroscopic object with a coarse-grained description having two states (e.g. a piece of iron magnetised along or against the $z$ axis). Otherwise, if the system merely implements a bit in a higher dimensional state space of dimension $d>2$, the internal degrees of freedom can be used as a memory and a bit-swap can be performed, as illustrated in Fig.~\hyperref[fig:bit_flip]{1a} for $d=3$. Thus, when we speak of classical systems, we take them to be fundamentally $d$-dimensional. And when we speak of memoryless dynamics, we mean probabilistic jumps between these discrete states occurring at rates independent of the system's past. Of course, this is just the standard setting for classical Markov processes. Then, for every finite dimension $d$, there exist processes that cannot be performed in a memoryless fashion.

This setting should also be contrasted with classical systems with a continuous phase space. Probably the simplest example is given by an isolated classical pendulum, i.e., a harmonic oscillator. The state of the system is then given by a point in (more generally:  distribution over) a two-dimensional phase space $[x,p]$ describing its position $x$ relative to the equilibrium position, and its momentum $p$. If we now identify the pendulum's state $[-x_0,0]$ with a bit state 0 and $[x_0,0]$ with a bit state 1, the time evolution of the system clearly performs a bit-swap operation within half a period. However, a pendulum is not a simple two-level system, but rather a system with a continuously infinite number of states. Thus, the bit-swap is performed by employing infinitely many ancillary memory states: $[-x_0,0]$ evolves through states with $x>-x_0$ and positive momenta $p$ to $[x_0,0]$, while $[x_0,0]$ evolves through states with $x<x_0$ and negative momenta $p$ to $[-x_0,0]$, and the momentum $p$ effectively acts as a register that carries the information about the past. For classical systems with a continuous phase space it is then difficult to properly assess the number of memory states used during a given evolution. Hence in this work we will only focus on discrete systems.

\begin{figure}[t]
	\centering
	\begin{tikzpicture}
	% External picture
	\node (myfirstpic) at (0,0) {\includegraphics[width=0.7\columnwidth]{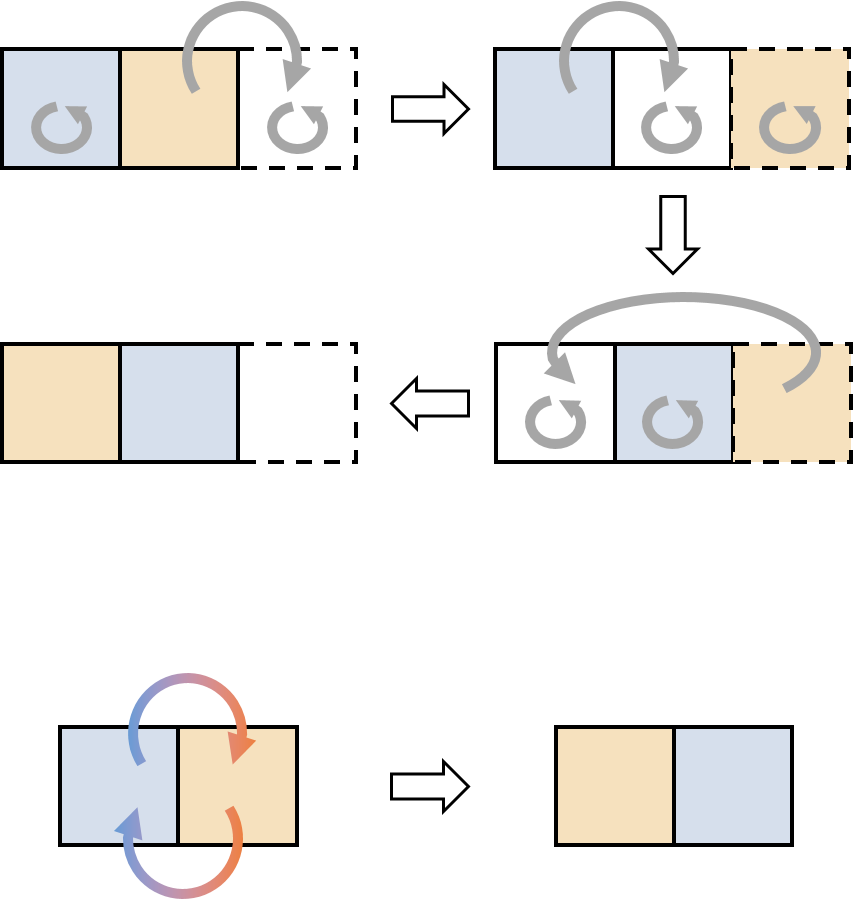}};
	% Labels
	\node at (-2.8,4) {\color{black} (a)};
	\node at (0,4) {\color{black} \textbf{Classical bit-swap}};
	\node at (0,3.6) {\color{black} 1 memory state, 3 time-steps};
	\node at (-2.8,-0.8) {\color{black} (b)};
	\node at (0,-0.8) {\color{black} \textbf{Quantum bit-swap}};
	\node at (0,-1.2) {\color{black} 0 memory states, 1 time-step};		
	\end{tikzpicture}
	\caption{\label{fig:bit_flip}\textbf{Space-time cost for classical and quantum bit-swap.} (a) Space-time optimal realisation of a bit-swap, i.e., a transposition between two states (solid line boxes), using one memory state (dashed line boxes) and three time-steps. Each time-step is composed of a continuous memoryless dynamics that does not affect one of the states, and maps the remaining two to one of them. (b) In the quantum regime, a bit-swap can be performed without any memory, simply by a time-continuous unitary process $\exp(i\sigma_x t)$ that continuously connects the identity operation at time $t=0$ with the bit-swap, represented by Pauli $x$ operator $\sigma_x$, at time $t=\pi/2$. During the process, the information about the initial state of the system is preserved in quantum coherence.}
\end{figure}
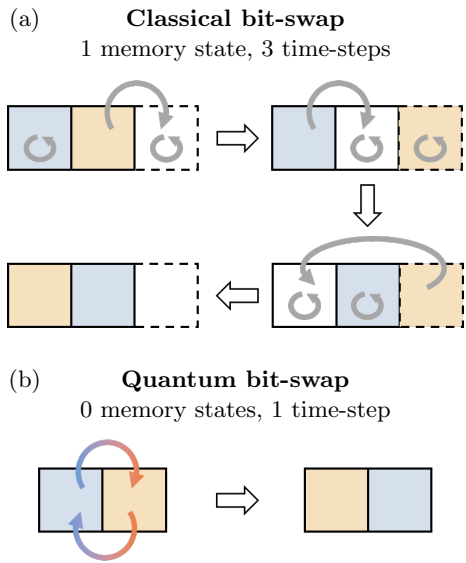

Ultimately we are then interested in a stochastic process in which discrete outputs $i$ are observed for given discrete inputs $j$. This process is characterized by a matrix of transition probabilities $P_{i|j}$. In the classical setting, we want to know whether there exists a classical memoryless dynamics (described by a Markovian master equation) involving these states that outputs $P_{i|j}$ after some time.  Quantum-mechanically, we are similarly asking whether a \emph{quantum} memoryless dynamics (described by a Markovian quantum master equation) can output $P_{i|j}$ after some time. In this work we highlight that these two questions admit very different answers, both in terms of which $P_{i|j}$ can arise from memoryless processes and in terms of the memory required to achieve a given~$P_{i|j}$. 
	
For concreteness, assume $P_{i|j}$ results from a thermalisation process, which typically satisfies the so-called \emph{detailed balance condition}. Physically, we are then asking whether the observed $P_{i|j}$ is compatible or not with a process involving no memory effects such as information backflows from the environment~\cite{breuer2009measure}. Classically, $P_{i|j}$ originates from incoherent jumps induced by interacting with the environment (absorbing or emitting energy). Alternatively, we can see classical dynamics as the evolution of a quantum system that undergoes very strong decoherence at all times, so that any non-classical effects are killed right away. In fact, also standard quantum thermalisation models (weak coupling with a very large thermal bath) can be understood in this way, since they are unable to generate quantum superpositions of energy states. As soon as we move away from this semiclassical limit, however, we see that more exotic thermalisation processes can generate $P_{i|j}$ that classically would necessarily signal memory effects, but that quantum mechanically can emerge from memoryless processes due to quantum coherence. 
		
%------------------------------------------------------------------
		 
\subsection{Summary of results}

In this work we identify three aspects of potential quantum advantage in simulating stochastic processes. First, in Sec.~\ref{sec:embeddable}, we investigate the possibility to \emph{simulate classical processes requiring memory using quantum memoryless dynamics}. More precisely, we compare the sets of all stochastic processes that can be generated by time-continuous memoryless dynamics in the classical and quantum domains~\footnote{Note that the term ``memoryless'' is used throughout the paper as a synonym of Markovian, i.e., that the evolution only depends on the current state of the system and not on its history. Such evolution may still require an auxiliary  clock system (used, e.g., to know how much longer the system should to be coupled to an external control field) and a counter system (that may be used to record the current channel in the sequence of channels necessary to implement the given dynamics). These constitute extra resources that one may want to separately account for, e.g. using the framework of quantum clocks (see, e.g., Ref.~\cite{erket2017autonomous}).}. We prove that the latter set is strictly larger than the former one, i.e., that there exist stochastic processes that classically require memory to be implemented, but can be realised by memoryless quantum dynamics. As an example, consider a random walk on a cyclic graph with three sites, where the walker can either move clockwise, anti-clockwise, or stay in place. As we present in Fig.~\ref{fig:embed_regions}, only a small orange subset of such walks can arise from a continuous classical evolution that does not employ memory (note that, differently from other investigations~\cite{gualtieri2020quantum}, we do not put any restriction on the classical dynamics beyond the fact that it is memoryless). However, if we allow for continuous memoryless quantum evolution, all stochastic processes in the much larger blue set can be achieved. Besides this particular class, in this work we provide general constructions for whole families of stochastic processes for any finite-dimensional systems that require memory classically, but can be implemented quantumly in a memoryless fashion.

Second, in Sec.~\ref{sec:space_time}, we go beyond the simple distinction between stochastic processes that can or cannot be simulated without memory, and take a more quantitative approach, thus investigating \emph{quantum memory advantages}. To this end, we employ the recent formalism of Ref.~\cite{wolpert2019space}, which allows one to quantify the classical space-time cost of a given stochastic process, i.e., the minimal amount of memory and time-steps needed to classically implement a given process. We extend this approach to the quantum domain in order to analyse the quantum space-time cost. An illustrative example is given by the bit-swap process presented in Fig.~\ref{fig:bit_flip}, which in the classical setting requires either one memory state and three time-steps, or two memory states and two time-steps. However, if one allows for quantum evolution, such a bit-swap can be performed in a continuous and memoryless fashion through a simple unitary evolution $\exp(i\sigma_x t)$ with $\sigma_x$ denoting the Pauli $x$ operator. More generally, the authors of Ref.~\cite{wolpert2019space} have characterised the space-time cost for the family of $\{0,1\}$-valued stochastic processes (i.e., all discrete functions). Their bound shows an unavoidable classical trade-off between the number of memory states $m$ and the number of time-steps $\tau$ needed to realise a given stochastic process on $N$ systems of dimension $d$. Crucially, a typical process necessarily requires extra resources, meaning that either $m$ or $\tau$ is exponential in~$N$. In this paper we prove that in the quantum regime all such processes can be simulated with zero memory states and in at most two time-steps, demonstrating an advantage over the best possible classical implementation.

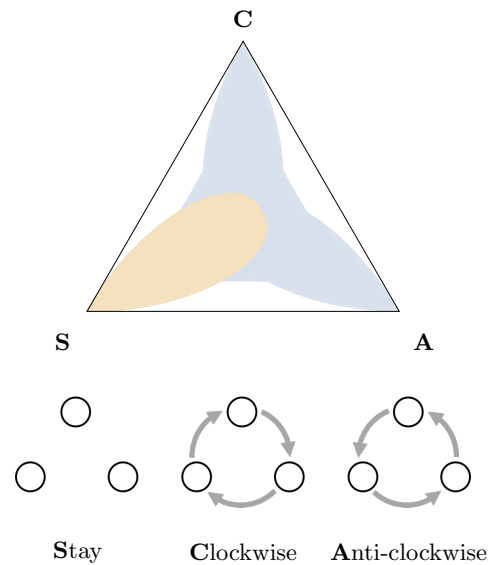
\begin{figure}[t]
	\centering
	\begin{tikzpicture}
		% External picture
		\node (myfirstpic) at (0,0) {\includegraphics[width=0.7\columnwidth]{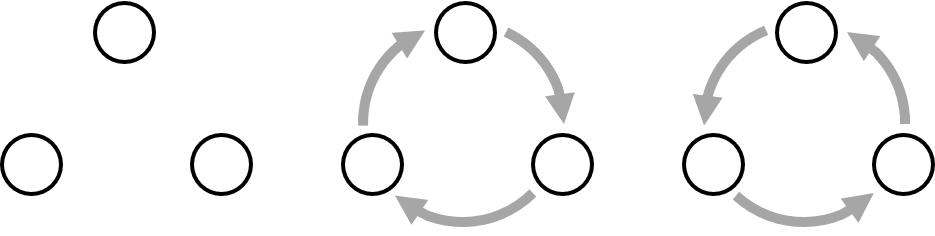}};
		% Labels
		\node at (0,-1.3) {\color{black} \textbf{C}lockwise};
		\node at (2.2,-1.3) {\color{black} \textbf{A}nti-clockwise};
		\node at (-2.2,-1.3) {\color{black} \textbf{S}tay};
		% External picture
		\node (myfirstpic) at (0,3.7) {\includegraphics[width=0.5\columnwidth]{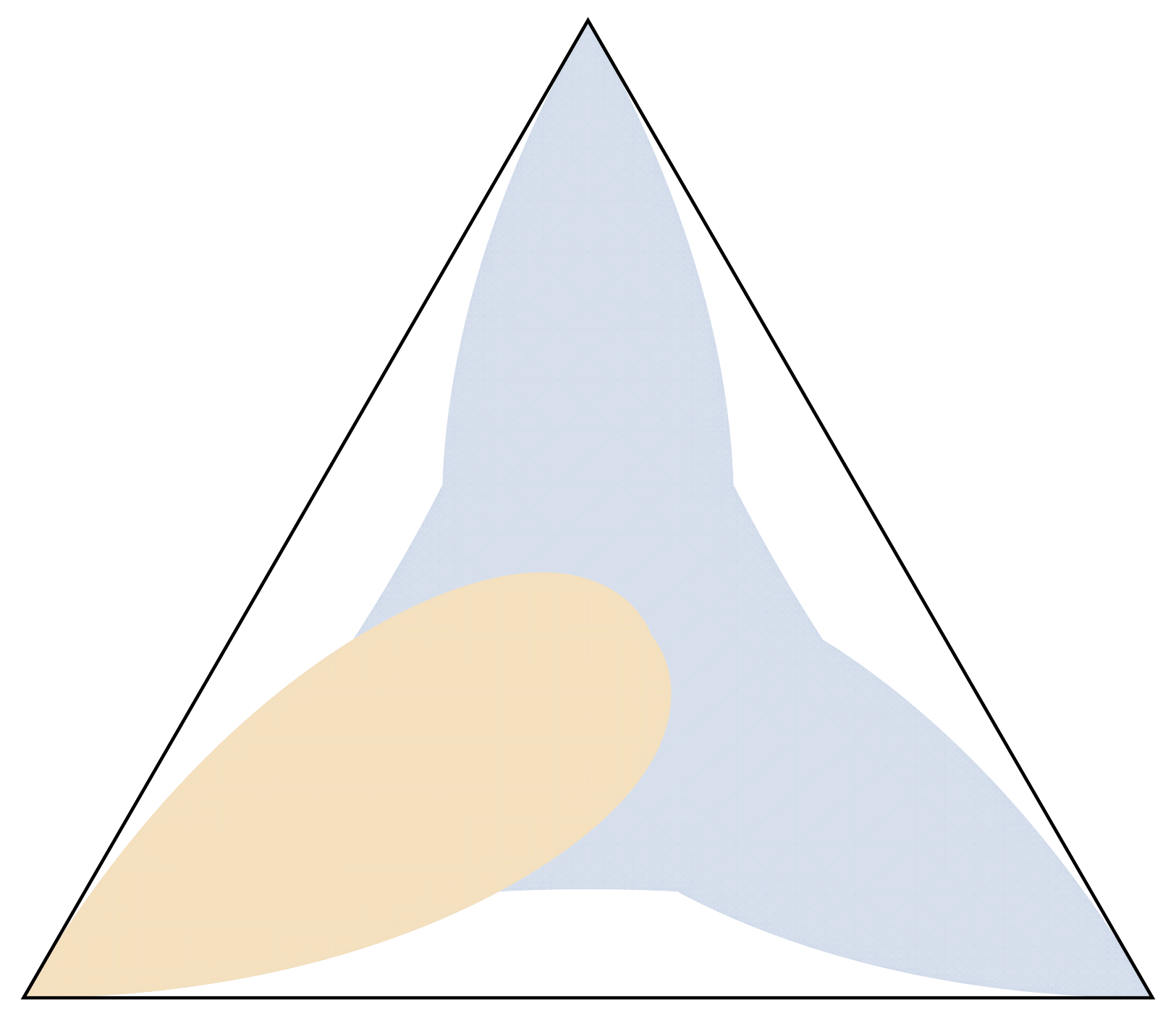}};	
		% Labels
		\node at (0,5.8) {\color{black} \textbf{C}};
		\node at (-2.4,1.5) {\color{black} \textbf{S}};
		\node at (2.4,1.5) {\color{black} \textbf{A}};
	\end{tikzpicture}
	\caption{\label{fig:embed_regions} \textbf{Classical vs quantum memoryless processes.} The vertices of the triangle correspond to deterministic processes (\textbf{S}: stay, \textbf{C}: move clockwise, \textbf{A}: move anti-clockwise) for a random walker moving between three states. Points inside the triangle correspond to probabilistic mixtures (convex combinations) of these three deterministic processes, e.g., the centre of the triangle corresponds to the maximally mixing dynamics (with \textbf{S}, \textbf{C} and \textbf{A} each happening with probability 1/3). The orange petal-shaped region contains all stochastic processes that can arise from time-continuous memoryless classical dynamics. For time-continuous memoryless quantum dynamics this set is enlarged by the remaining shaded region in blue. For details see Sec.~\ref{sec:examples} and, in particular, Fig.~\ref{fig:3x3_regions}.}
\end{figure}

Third, in Sec.~\ref{sec:state} we study \emph{memory advantages in control} by comparing classical and quantum continuous memoryless dynamics in terms of the set of accessible final states. We assume a fixed point of the evolution is given, which is a realistic physical constraint in dissipative processes and typically, but not necessarily, coincides with the thermal Gibbs state. A standard example is given by a thermalisation of the system to the environmental temperature. Here, we employ our recent result~\cite{lostaglio2020fundamental}  characterising the input-output relations of classical Markovian master equations with given fixed point. We show how quantum memoryless dynamics \emph{with the same fixed point} allow one to access a larger set of final states. This is most evident in the case of maximally mixed fixed points (corresponding to the environment in the infinite temperature limit), since every transformation that is classically possible with arbitrary amounts of memory can be realised in a memoryless fashion in the quantum domain. For general fixed points, we prove that an analogous result holds for systems of dimension $d=2$, and argue that the set of accessible states is strictly larger in the quantum regime than in the classical one for all $d$. Since it is known that memory effects enhance cooling~\cite{alhambra2019heat, taranto2020exponential}, a direct consequence of our results is that quantumly it is possible to bring the two-dimensional system below the environmental temperature without employing memory effects, something that is impossible classically (see Fig.~\ref{fig:cooling}).

Finally, in Sec.~\ref{sec:applications} we discuss the potential for practical applications of our results, while Sec.~\ref{sec:outlook} contains the outlook for future research.

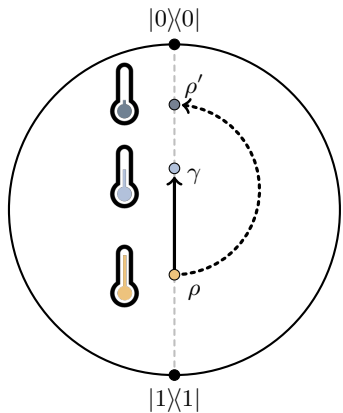
\begin{figure}[t]
	\centering
	\begin{tikzpicture}[line cap=round,line join=round,x=2.2cm,y=2.2cm]
		\clip(-1.15130521264221,-1.2346670851205752) rectangle (1.2120740854657621,1.3343488940797);
		% Thermometers	
		\node at (-0.3,0.7) {\includegraphics[width=0.9cm]{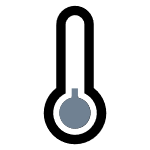}};
		\node at (-0.3,0.2) {\includegraphics[width=0.9cm]{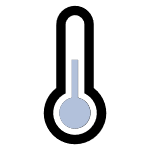}};
		\node at (-0.3,-0.4) {\includegraphics[width=0.9cm]{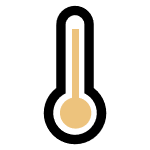}};
		% Bloch sphere		
		\draw [line width=0.8pt] (0.,0.) circle (2.2cm);
		\draw [line width=0.8pt,color=lightgray,dash pattern=on 2pt off 2pt] (0.,1.)-- (0.,-1.);
		% Paths with arrows
		\draw [->,line width=1.2pt,color=black,dotted,domain=-90:85] plot({1.1319436901052642*cos(\x)/2.2},{0.12137003521531087+1.1319436901052642*sin(\x)/2.2});
		\draw [->,line width=1.2pt,color=black] (0.,0.12137003521531087-1.1319436901052642/2.2)-- (0.,0.2);
		% States
		\draw [fill=tikzOrange] (0,0.12137003521531087-1.1319436901052642/2.2) circle (2.0pt);
		\draw [fill=tikzBlue] (0.,0.25) circle (2.0pt);
		\draw [fill=tikzBlue2] (0,0.12137003521531087+1.1319436901052642/2.2) circle (2.0pt);
		\draw [fill=black] (0.,1.) circle (2.0pt);
		\draw [fill=black] (0.,-1.) circle (2.0pt);
		% State labels
		\node at (0,1.15) {\color{black} $\ketbra{0}{0}$};
		\node at (0,-1.15) {\color{black} $\ketbra{1}{1}$};
		\node at (0.12,0.2) {\color{black} $\gamma$};
		\node at (0.12,-0.5) {\color{black} $\rho$};
		\node at (0.12,0.75) {\color{black} $\rho'$};
	\end{tikzpicture}
	\caption{\label{fig:cooling}\textbf{Markovian cooling of a qubit.} Classical memoryless processes can only cool the initial state $\rho$ of a two-dimensional system to the thermal state $\gamma$ at the environmental temperature (path along the solid line arrow). Quantum memoryless dynamics with fixed point $\gamma$ allows one to cool the system below that, all the way to the state $\rho'$ with the lowest temperature achievable by classical processes with memory (path along the dotted line arrow). For details see Sec.~\ref{sec:qubit_ext} and, in particular, Fig.~\ref{fig:qubit_path}.}
\end{figure}

%------------------------------------------------------------------

\section{Embeddability of stochastic processes}
\label{sec:embeddable}

\subsection{Classical embeddability}

Given a discrete state space, $\{1,\dots,d\}$, the state of a finite-dimensional classical system is described by a probability distribution $\v{p}$ over these states. A \emph{stochastic matrix or process} $P$ is a matrix $P_{i|j}$ of transition probabilities,
\begin{equation}
	P_{i|j} \geq 0,\quad \sum_{i} P_{i|j} =1,
\end{equation}
which describes the evolution of the system from one state $\v{p}$ to another $P\v{p}$. 
	
Classically, the $P_{i|j}$ that can be achieved without employing memory are known as \emph{embeddable}. A stochastic matrix $P$ is embeddable if it can be generated by a continuous Markov process~\cite{davies2010embeddable}. This notion can be understood as a control problem involving a master equation. Namely, introducing a \emph{rate matrix} or \emph{generator} $L$ as a matrix with finite entries satisfying 
\begin{equation}
	L_{i|j} \geq 0\mathrm{~for~}i \neq j, \quad\sum_{i} L_{i|j} =0,
\end{equation}
a continuous one-parameter family $L(t)$ of rate matrices generates a family of stochastic processes $P(t)$ satisfying
\begin{equation}
\label{eq:P}
	\frac{d}{dt} P(t) = L(t) P(t),\quad P(0)= \iden.
\end{equation}
The aim of the control $L(t)$ is to realize a target stochastic process $P$ at some final time $t_f$ as $P=P(t_f)$. If this is possible for some choice of $L(t)$, then $P$ is embeddable; and if there exists a time-independent generator $L$ such that $P = e^{L t_f}$, then we say that $P$ can be embedded by a time-homogeneous Markov process. A final technical comment is that we also consider the case $t_f = \infty$ to be embeddable (in Ref.~\cite{wolpert2019space} this case was referred to as limit-embeddable). Then, $P$ cannot be generated in any finite time, but it can be approximated arbitrarily well. This is the case, e.g., with the bit erasure process: \mbox{$0 \mapsto 0$}, \mbox{$1 \mapsto 0$}~\cite{wolpert2019space}.

The question of which stochastic matrices $P$ are embeddable is a challenging open problem that has been extensively investigated for decades~\cite{elfving1937theorie,kingman1962imbedding,runnenburg1962elfving,goodman1970intrinsic,carette1995characterizations,davies2010embeddable}. The full characterization does not go beyond $2\times 2$ and $3\times 3$ stochastic matrices, however various necessary conditions have been found. In particular, in Ref.~\cite{goodman1970intrinsic} it was proven that every embeddable stochastic matrix $P$ satisfies the following inequalities:
\begin{align}
	\label{eq:necessaryembeddable}
	\prod_i P_{i|i} \geq \det P \geq 0.
\end{align}
The condition $\det P \geq 0$ is, in fact, also known to be sufficient in dimension $d=2$~\cite{kingman1962imbedding}, and a time-independent rate matrix $L$ can then be found.

\begin{ex}[Thermalisation]
	\label{ex:classicalthermalization}
	Consider a two-level system with energy gap $E$ incoherently exchanging energy with a large environment at inverse temperature $\beta$. Assume $ \beta E =1$. Suppose one observes $P_{i|j}$ satisfying the detailed balance condition: \mbox{$P_{1|0} = P_{0|1}e^{-1}$}. Can this stochastic process originate from a memoryless dynamics? Using the condition in Eq.~\eqref{eq:necessaryembeddable} one can verify this is the case if and only if \mbox{$P_{0|1} \leq e/(1+e) \approx 0.731$}. The dynamics that realises $P$ is then a standard thermalisation process whereby the system's state $\v{p}$ exponentially relaxes to the equilibrium distribution $\v{\gamma} = (e/(1+e), 1/(1+e))$, according to the classical master equation:
	\begin{equation}
	\frac{d}{dt}\v{p}(t) = R( \v{\gamma} - \v{p}(t)),
	\end{equation} 
	where $R$ denotes the thermalisation rate. The stochastic process $P$ is realized by a partial thermalisation lasting for a time $t = - \log (1- \frac{1+e}{e} P_{0|1})/R$.  Intuitively the de-excitation probability $P_{0|1}$ cannot be made larger than $e/(1+e)$ because memoryless thermalisations need to satisfy detailed balance at every intermediate timestep (and so there is always some probability of absorbing an excitation from the bath).  In fact, $P_{0|1} > e/(1+e)$ can only be realized if memory effects are present. An example of such a process is a detailed balanced $P$ with $P_{0|1} = 1$, which is the ``$\beta$-swap'' providing enhanced cooling in Ref.~\cite{alhambra2019heat}. Indeed the latter process can be approximated by a Jaynes-Cummings coupling of a two-level system to an ``environment'' given by a single harmonic oscillator initialized in a thermal state, which gives rise to a highly non-Markovian evolution of the system.
\end{ex}

%------------------------------------------------------------------

\subsection{Quantum embeddability}

A state of a finite-dimensional quantum system is given by a density operator $\rho$, i.e. a positive semi-definite operator with trace one that acts on a $d$-dimensional Hilbert space $\H_d$. A general evolution of a density matrix is described by a quantum channel $\E$, which is a completely positive trace-preserving map from the space of density matrices to itself. Now, focussing on the computational basis $\{\ket{k}\}_{k=1}^d$ of $\H_d$, suppose we input the quantum state $\rho_{\v{p}} = \sum_k p_k \ketbra{k}{k}$, apply the channel $\mathcal{E}$ and measure the resulting state $\E(\rho_{\v{p}})$ in the computational basis. The measurement outcomes will be distributed according to $P \v{p}$, where
\begin{align}
	\label{eq:qc_map}
	P_{i|j}=& \matrixel{i}{\mathcal{E}\left(\ketbra{j}{j}\right)}{i}.
\end{align}
In this way, the preparation of $\rho_{\v{p}}$, followed by a channel $\mathcal{E}$ and the computational basis measurement, simulates the action of a stochastic process $P$ on the classical state~$\v{p}$.  

Surprisingly, as far as we are aware the set of processes $P$ that can be simulated by a quantum process without employing memory has not been named or studied before. Hence, we define a stochastic matrix $P$ as \emph{quantum-embeddable} if it can be simulated by a quantum process as in Eq.~\eqref{eq:qc_map} with $\mathcal{E}$ a \emph{Markovian quantum channel}~\cite{wolf2008assessing}, i.e. a channel that can result from a Markovian master equation (the quantum analogue of Eq.~\eqref{eq:P}). 

Here, we describe what it rigorously means for a quantum channel $\mathcal{E}$ to be Markovian. Despite the difference in jargon between the two communities, Markovianity for channels is the quantum analogue of the classical notion of embeddability and captures the fact that $\mathcal{E}$ can be realized without employing memory effects. It can also be understood as a control problem, but this time involving a quantum master equation. More precisely, the rate matrix $L$ is replaced by  a \emph{Lindbladian}~\cite{gorini1976completely,lindblad1976generators}, which is a superoperator $\mathcal{L}$ acting on density operators and satisfying 
\begin{equation}
	\label{eq:lindbladian}
	\mathcal{L}(\cdot) =  -i[H,\cdot]+\Phi(\cdot) - \frac{1}{2}\{ \Phi^*(\iden),\cdot \},
\end{equation}
with the first term describing unitary evolution and the remaining ones encoding the dissipative dynamics, e.g., due to the interaction with an external environment. Here $H$ is a Hermitian operator, $[A,B]:= AB-BA$ denotes a commutator, $\Phi$ is a completely positive superoperator, $\Phi^*$ denotes the dual of $\Phi$ under the Hilbert-Schmidt scalar product, and $\{A,B\}:= AB+BA$ stands for the anticommutator. In analogy with Eq.~\eqref{eq:P}, a continuous one-parameter family of Lindbladians $\mathcal{L}(t)$ generates a family of quantum channels $\mathcal{E}(t)$ satisfying
\begin{equation}
	\frac{d}{dt}\hat{\mathcal{E}}(t) = \hat{\mathcal{L}}(t)\hat{\mathcal{E}}(t),\quad \hat{\mathcal{E}}(0)=\hat{\mathcal{I}},
\end{equation} 
where hats indicate superoperators (i.e., matrix representations of quantum channels that act on vectorised quantum states) and $\mathcal{I}$ denotes the identity channel. A quantum channel $\mathcal{E}$ is Markovian~\cite{wolf2008assessing} if $\mathcal{E} = \mathcal{E}(t_f)$ for some choice of the Lindbladian $\mathcal{L}(t)$ and $t_f$ (perhaps $t_f = +\infty$). In other words, $\mathcal{E}$ is Markovian if it is a channel that results from integrating a quantum master equation up to some time $t_f$. Any given Markovian channel $\mathcal{E}$ gives a stochastic process $P$ through Eq.~\eqref{eq:qc_map}. The aim of the control $\mathcal{L}(t)$ is to achieve a target stochastic matrix $P$ after some time $t_f$. More formally we introduce the following definition.

\begin{defn}[Quantum-embeddable stochastic matrix]
	A stochastic matrix $P$ is \emph{quantum-embeddable} if
	\begin{equation}
		\label{eq:q_embed}
		P_{i|j} =\bra{i}  \mathcal{E}\left(\ketbra{j}{j}\right)\ket{i},
	\end{equation}	
	where $\mathcal{E}$ is a Markovian quantum channel.
\end{defn}

\begin{ex}
	\label{ex:tpm}
	Consider a two-point projective measurement scheme (TPM)~\cite{campisi2016colloquium}. First, a projective energy measurement is performed and one finds the system in a well defined energy state $j$. Then, a quantum evolution $\mathcal{E}$ follows, and finally a second energy measurement returns the outcome $i$ with probability $P_{i|j}$. We can then ask whether it is possible that the process $\mathcal{E}$ generating $P_{i|j}$ resulted from a Markovian quantum master equation. Suppose one observes
	\begin{equation}
	\label{eq:unistochasticexample}
	P= \begin{bmatrix}
	1/3 & 2/3 \\
	2/3 & 1/3 
	\end{bmatrix}.
	\end{equation}
	Then, it is straightforward to show that the above can arise from the following unitary dynamics $U$ (which is a Markovian channel) 
	\begin{equation}
	\label{eq:unistochasticexample2}
	U= \begin{bmatrix}
	\sqrt{1/3} & \sqrt{2/3} \\
	\sqrt{2/3} & -\sqrt{1/3}
	\end{bmatrix},
	\end{equation}
	while according to Eq.~\eqref{eq:necessaryembeddable} it is impossible to generate such a $P$ using classical memoryless dynamics.
\end{ex}

%------------------------------------------------------------------

\subsection{Quantum advantage}
\label{sec:examples}

One can easily see that all (classically) embeddable stochastic processes are also quantum-embeddable: given a classical generator $L$ one chooses the CP map $\Phi$ defining the Lindbladian $\L$ in Eq.~\eqref{eq:lindbladian} to be
\begin{equation}
	\Phi(\cdot)=\sum_{ij} K_{ij}(\cdot) K_{ij}^\dagger,\quad K_{ij}=\sqrt{L_{i|j}} \ketbra{i}{j}.
\end{equation}
However, the converse is not true. There exist many stochastic matrices $P$ that can be generated by a quantum, but not a classical Markov process. The simplest example is given by a non-trivial permutation $\Pi$, satisfying 
\begin{equation}
	\det \Pi = \pm 1,\quad \prod_i \Pi_{i|i} = 0.
\end{equation}
Clearly, Eq.~\eqref{eq:necessaryembeddable} is violated and hence $\Pi$ is not embeddable. However, noting that every unitary channel $U(\cdot)U^\dagger$ is Markovian (by choosing the Lindbladian with no dissipative part and $H$ such that $U=\exp(iHt_f)$), and that a permutation matrix $\Pi$ is unitary, we conclude that every permutation $\Pi$ is quantum-embeddable. This conclusion also proves that neither of the two conditions in Eq.~\eqref{eq:necessaryembeddable} are necessary for quantum-embeddability. 

More generally, a larger class of stochastic matrices that are quantum-embeddable is given by the set of \emph{unistochastic} matrices~\cite{bengtsson2004importance,bengtsson2005birkhoff}. These are defined as all stochastic matrices $P$ satisfying 
\begin{equation}
	P_{i|j} = |\bra{j}U\ket{i}|^2
\end{equation}
for some unitary matrix $U$, and the argument for quantum embeddability is analogous to the one given for permutation matrices. The set of unistochastic matrices includes permutations, but also other (classically) non-embeddable stochastic matrices. As an example one can consider the bistochastic matrix $P$ in Eq.~\eqref{eq:unistochasticexample} or, in fact, any other $2 \times 2$ bistochastic matrix. This is because in dimension $d=2$ every bistochastic matrix is unistochastic and so it is quantum-embeddable.

Beyond these examples we prove a simple general result that allows one to find larger families of quantum embeddable stochastic matrices.

\begin{lem}[Monoid property]
	\label{lem:composition}
	The set of quantum-embeddable stochastic matrices contains identity and is closed under composition, i.e., if $P$ and $Q$ are quantum-embeddable, then also $PQ$ is.
\end{lem}
\begin{proof}
First, identity is obviously quantum-embeddable as it arises from a trivial Lindbladian $\L=0$. Now, note that the composition of Markovian quantum channels gives a Markovian quantum channel. Next, notice that a completely dephasing map
\begin{equation}
	\D(\cdot):=\sum_k \matrixel{k}{\cdot}{k} \ketbra{k}{k}
\end{equation}
is a Markovian quantum channel. Finally, the composition $\mathcal{E} = \E_P\circ\D\circ \E_Q$, with $\E_P$ and $\E_Q$ being quantum channels describing the quantum embeddings of $P$ and~$Q$, is a Markovian quantum channel which quantum-embeds the stochastic process described by $PQ$. 
\end{proof}

Let us now discuss the consequences of Lemma~\ref{lem:composition} with increasing generality. We start with the following corollary for dimension $d=2$.

\begin{cor}
	\label{cor:2x2}
	All $2\times 2$ stochastic matrices are quantum embeddable.
\end{cor}
\begin{proof}
A general $2\times 2$ stochastic matrix $P$ can be written as 
\begin{equation}
	P= \begin{bmatrix}
	a & 1-b \\
	1-a & b 
\end{bmatrix}.
\end{equation}
If $\det P \geq 0$, then $P$ is embeddable and hence quantum embeddable. Otherwise, if $\det P < 0$, we can write $P=\Pi P'$ with $\Pi$ denoting the non-trivial $2\times 2$ permutation and $P'$ being a stochastic matrix with $\det P' \geq 0$. Since $P$ can be written as a composition of two quantum embeddable maps, by Lemma~\ref{lem:composition} it is also quantum embeddable.
\end{proof}

\begin{figure}[t]
	\includegraphics[width=0.67\columnwidth]{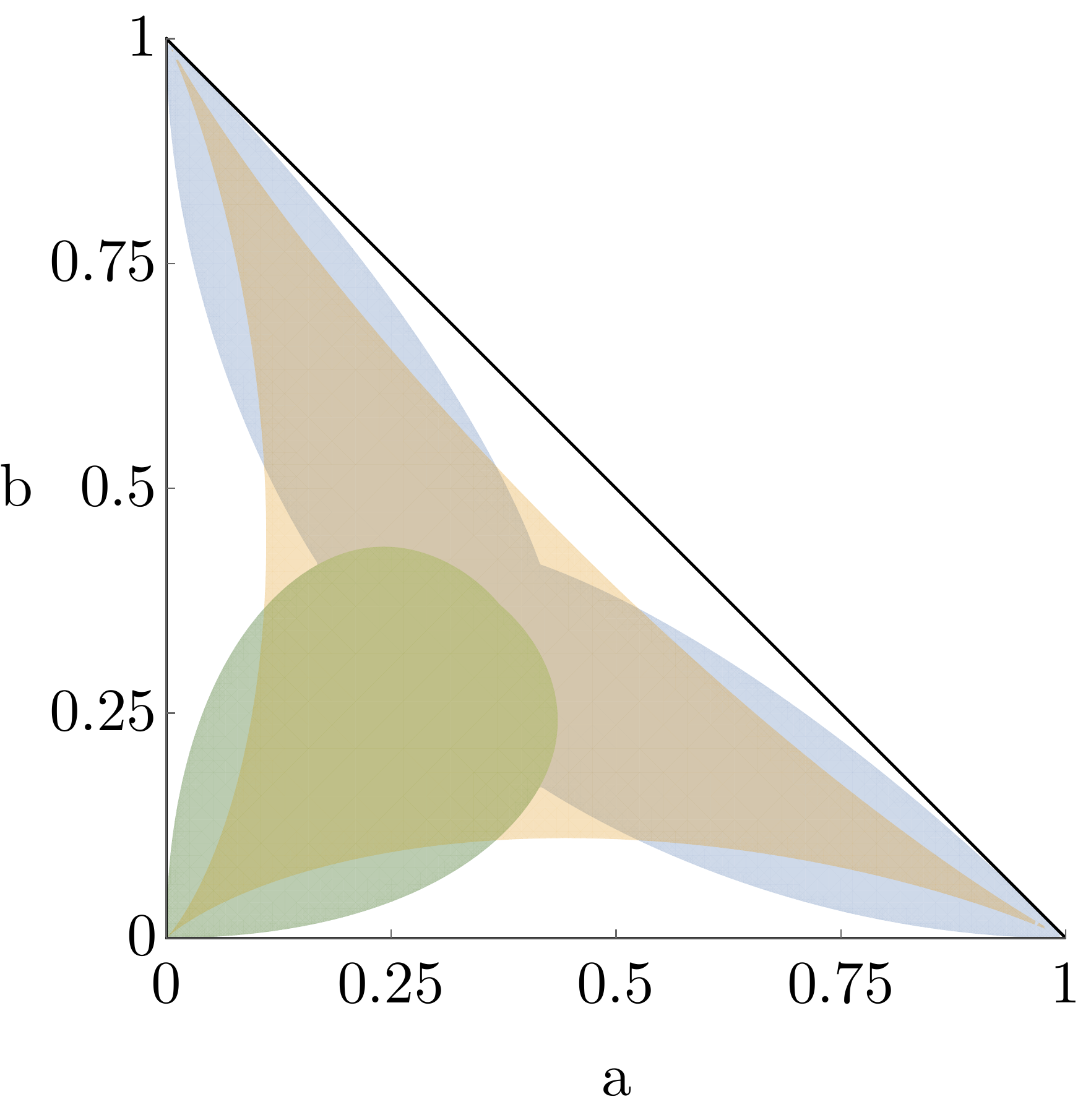}
	\caption{\label{fig:3x3_regions}\textbf{Embeddability of $3\times 3$ circulant matrices.} Every $3\times 3$ circulant matrix corresponds to a point within a half-square in the parameter space $[a,b]$ according to Eq.~\eqref{eq:circulant}. The green petal-shaped region around the origin contains all (classically) embeddable matrices. The set of quantum-embeddable matrices is larger and contains the orange triangle-like region of unistochastic matrices, as well as two blue petal-like regions corresponding to permutations of the classically embeddable region.}
\end{figure}

For $d\geq 3$ comparing quantum and classical embeddability becomes complicated due to the lack of a complete characterisation of classical embeddability. One can, however, focus on certain subclasses of stochastic processes that are better understood. For example, for the family of $3\times 3$ \emph{circulant stochastic matrices}, defined by
\begin{align}
	\label{eq:circulant}
	P= \begin{bmatrix}
	1-a-b & a & b \\
	b & 1-a-b & a \\
	a & b & 1-a-b
	\end{bmatrix},
\end{align}
the necessary and sufficient conditions for (classical) embeddability are known. Denoting the eigenvalues of $P$ by $\lambda_k = r_k e^{i \theta_k}$ with $\theta_k \in [-\pi,\pi]$, these conditions are given by~\cite{lencastre2016empirical}
\begin{equation}
	\forall k:~r_k \leq \exp \left[-\theta_k \tan(\pi/3) \right].
\end{equation}
We illustrate the set of classically embeddable circulant matrices by the green region in parameter space $[a,b]$ in Fig.~\ref{fig:3x3_regions}. On the other hand, due to Lemma~\ref{lem:composition}, quantum embeddable circulant stochastic matrices also include permutations of $P$, i.e., $\Pi P$ with $\Pi$ denoting circulant $3\times 3$ permutation matrices. In fact, this set contains not only permutations of $P$ but also compositions of $P$ with any unistochastic matrix; however, numerical verification suggests that this does not further expand the investigated set. As a result, the set of quantum embeddable stochastic matrices in the parameter space $[a,b]$ contains not only the region corresponding to classically embeddable matrices, but also its two copies (corresponding to two permutations), which we illustrate in blue in Fig.~\ref{fig:3x3_regions}. Moreover, all unistochastic circulant matrices (which are fully characterised by the ``chain-links'' conditions from Ref.~\cite{bengtsson2005birkhoff}) are also quantum embeddable. The resulting region is also plotted in Fig.~\ref{fig:3x3_regions} in orange. We thus clearly see that the set of quantum-embeddable stochastic circulant matrices is much larger than the classically embeddable one, since it contains the union of green, blue and orange regions. However, we do not expect that all $3\times 3$ circulant matrices are quantum embeddable, with the case $a=b=1/2$ being intuitively the least likely to arise from quantum Markovian dynamics. 

For general dimension $d$ one can generate families of quantum-embeddable matrices using Lemma~\ref{lem:composition} in an analogous way, by composing classically embeddable matrices with unistochastic ones. Moreover, employing Corollary~\ref{cor:2x2} we note that the set of quantum embeddable stochastic matrices also includes all matrices $P$ that can be written as products of \emph{elementary stochastic matrices} $P_{e_i}$ (also known as \emph{pinching matrices}), i.e., 
\begin{equation}
	P= P_{e_n} \dots P_{e_1}, \quad P_{e_i} = \Pi_i (P_2 \oplus I_{d-2}) \Pi_i,
\end{equation}
where $P_2$ is a general $2\times 2$ stochastic matrix, $I_{d-2}$ denotes identity on the remaining states, and $\Pi_i$ is an arbitrary permutation. Notably, for $d \geq 4$ such products of elementary matrices contain matrices that are not unistochastic~\cite{poon1987inclusion} and hence cannot be reduced to the examples above. In conclusion, quantum embeddings allow one to achieve many stochastic processes that necessarily require memory from a classical standpoint.

%------------------------------------------------------------------

\subsection{Discussion} 

From a physical perspective, it is now natural to ask: why the set of quantum embeddable stochastic matrices is strictly larger than the set of classically embeddable ones? To address this question, let us first consider the simple example of classically non-embeddable permutation matrices,
\begin{align}
	\Pi_{m} = \sum_{n=1}^d \ketbra{n \oplus m}{n}, \quad m = 1,\dots,d-1,
\end{align}
where $\oplus$ here denotes addition modulo $d$. A direct calculation shows that $\Pi_{m} = e^{i H m}$ with Hamiltonian 
\begin{equation}
	H = \sum_{n=1}^d \frac{2\pi (n-1)}{d} \ketbra{\psi_n}{\psi_n},
\end{equation}
and
\begin{equation}
	\ket{\psi_n}=\frac{1}{\sqrt{d}}\sum_{k=1}^d e^{-i2\pi (k-1)(n-1)/d} \ket{k}.
\end{equation}
We thus see that the continuous and memoryless Hamiltonian evolution creates a superposition of classical states $\ket{n}$ on the way between identity and $\Pi_m$. The intuitive picture that emerges is that the quantum superposition between classical states created  during the evolution effectively acts as a memory. For example, when we perform a rotation of the Bloch sphere around the $y$ axis, we can implement a bit-swap sending $\ket{0}$ to $\ket{1}$ and vice versa, but the path the state follows (going through $\ket{+}$ if the initial state was $\ket{0}$, and through $\ket{-}$ if the initial state was $\ket{1}$) will preserve the memory about the initial state. At the same time, a classical memoryless process moving $(1,0)$ towards $(0,1)$ and $(0,1)$ towards $(1,0)$ cannot proceed beyond the point at which the two trajectories meet. 

One might wonder if we can quantify the coherent resources required for the advantage. There are various frameworks that have been put forward to quantify superposition (the resource	theory of coherence in its various forms~\cite{streltsov2017colloquium} or that of asymmetry~\cite{marvian2016quantify}). However, none of them associates costs to permutations (technically, these are ``free operations'') despite the fact that they carry an advantage in our setting. What we allude to in the present discussion is that, while these theories assign no cost to these operations, they can be performed in a Markovian fashion only because one can continuously connect different basis states through the creation of superpositions. Thus, current frameworks seem inadequate to capture the resources involved in the quantum memory advantage. An alternative framework would have to quantify the maximum amount of coherence that must be created at \emph{intermediate times}, minimised over all Markovian realisations of a target channel. This quantity may then be given an operational meaning in terms of minimal coherent resources one must input to realise the corresponding protocol. We leave this research direction for future work.

Our results can also be naturally connected to a result by Montina~\cite{montina2008exponential}, who proved that Markovian hidden variable models reproducing quantum mechanical predictions necessarily require a number of continuous variables that grows linearly with the Hilbert space dimension, and hence exponentially with the system size. Intuitively, here we are showing that this ``excess baggage''~\cite{hardy2004quantum} can be exploited to simulate memory effects. In what follows we provide a quantification of the advantage beyond the embeddable/non-embeddable dichotomy. We will see that quantum theory allows for advantages in the simulation of stochastic processes by memoryless dynamics.    

%------------------------------------------------------------------

\section{Space-time cost of a stochastic process}
\label{sec:space_time}

In this section we first recall a recently introduced framework for the quantification of the space and time costs of simulating a stochastic process by memoryless dynamics~\cite{wolpert2019space}. We then extend it to the quantum domain and prove a quantum advantage in the corresponding costs.

\subsection{Classical space-time cost}

Let $P$ be a non-embeddable stochastic matrix acting on $d$ so-called \emph{visible} states. We then want to ask: how many additional \emph{memory} states $m$ does one need to add, in order to implement $P$ by a classical Markov process? Formally, one looks for an embeddable stochastic matrix $Q$ acting on $d+m$ states whose restriction to the first $d$ rows and columns is identical to $P$. When this happens, $Q$ is said to \emph{implement} $P$ with $m$ memory states. In fact, given any $d$-dimensional distribution $\v{p}$, if we take the $d+m$ dimensional distribution $\v{q} = (\v{p}, 0 \dots 0)$, then $Q \v{q} = (P\v{p},0,\dots 0)$. Following Ref.~\cite{wolpert2019space} we now have the following.

\begin{defn}[Space cost]
	\label{def:space_cost}
	The \emph{space cost} of a $d \times d$ stochastic matrix $P$, denoted $C_{\rm space}(P)$, is the minimum $m$ such that a $(d+m) \times (d+m)$ embeddable matrix $Q$ implements $P$.
\end{defn}

\noindent As a technical comment we note that the above definition can be extended to situations in which visible and memory states are not disjoint, e.g., when the visible states on which $P$ acts are logical states defined by a coarse graining of the states on which $Q$ acts. Since this does not change any of the results presented here, we refer to Ref.~\cite{wolpert2019space} for further details and adopt the simpler definition given here. 

Once we find a matrix $Q$ that implements $P$, the next question is: what is the number of time-steps required to realise $Q$? The notion of a time-step is meant to capture the number of independent controls that are needed to achieve $Q$. A natural definition would be that the number of time-steps necessary to realise an embeddable stochastic matrix $Q$ is the minimum number $n$ such that
\begin{equation}
	\label{eq:sequencecontrols}
	Q =  e^{L^{(n)} t_{n}} \cdots e^{L^{(1)} t_{1}},
\end{equation}
where $L^{(1)}, \dots, L^{(n)}$ are time-independent generators, i.e. each $L^{(k)}$ is a control applied for some time $t_{k}$. This captures the idea of a sequence of autonomous steps, with $n$ the number of times an active intervention is required to ``quench'' $L^{(i)}$ to $L^{(i+1)}$. If $n=1$ one has an almost autonomous protocol, i.e. the only requirement is the ability to switch off the controls after time~$t_1$.
	
From a physical point of view, the issue with this definition is that it assigns an infinite cost to any realistic protocol in which controls are switched on and off in a continuous fashion. For example, suppose a single two-level system is kept in contact with an idealized dissipative environment while we slowly tune its energy gap. Such a protocol can be seen as the continuous limit of a sequence of steps described in Eq.~\eqref{eq:sequencecontrols}. As such, it would be assigned an infinite time cost according to the above definition, even though it is certainly experimentally feasible. To overcome this issue, note that by Levy's lemma~\cite{freedman1983approximating}, a crucial property of each step in the sequence of Eq.~\eqref{eq:sequencecontrols} is that the set of non-zero transition probabilities does not change. This naturally suggests to define the time cost as the number of times the set of ``blocked'' transitions changes. One could physically motivated this definition as follows. Consider a system interacting with a large thermal environment. Transitions between any pairs of system's energy levels are possible by absorption or emission of the corresponding energy from or to the bath. Absorptions are exponentially suppressed in the energy gap. To selectively couple only certain levels we either need to raise and lower infinite energy barriers~\cite{wolpert2019space}, or we need to engineer the spectrum of the environment so that only certain transitions can occur. Changing the set of system's energy levels involved in the interaction, by decoupling some and coupling new ones, is then a non-trivial control operation and we hence assign a cost to it. The time cost, defined as the number of times we need to change the set of coupled energy levels, is then a good proxy for the level of required control. It solves the issue with the previous definition and, in particular, it assigns a cost $n=1$ to the qubit protocol mentioned above.

To sum up, these considerations lead to the following definition of a one-step process~\cite{wolpert2019space}:

\begin{defn}[One-step process]
	A stochastic matrix $T$ is called \emph{one-step} if
	\begin{enumerate}
	\item it is embeddable;
	\item the controls $L(t)$ that generate $T$ at time $t_f$ through Eq.~\eqref{eq:P}  can be chosen such that the set of non-zero transition probabilities of $P(t)$ is the same for all $t \in (0,t_f)$.
	\end{enumerate}
\end{defn}

Putting all this together we obtain the notion of time cost from Ref.~\cite{wolpert2019space}:

\begin{defn}[Time cost]
	\label{def:time_cost}
	The \emph{time cost} $C_{\rm time}(P,m)$ of a $d \times d$ stochastic matrix $P$, while allowing for $m$ memory states, is the minimum number $\tau$ of one-step stochastic matrices $T^{(i)}$ of dimension $(d+m) \times (d+m)$ such that $Q =  T^{(\tau)} \cdots  T^{(1)}$
	implements $P$.
\end{defn}

%------------------------------------------------------------------

\subsection{Quantum space-time cost}

The framework presented above allows one to quantify the memory and time costs of implementing a given stochastic process by classical master equations. We now introduce a natural extension of the above to the quantum domain.

\begin{defn}[Quantum space cost]
	The \emph{quantum space cost} of a $d \times d$ stochastic matrix $P$, denoted $Q_{\rm space}(P)$, is the minimum $m$ such that the $(d+m) \times (d+m)$ quantum embeddable matrix $Q$ implements $P$.
\end{defn}

Concerning the definition of the quantum time cost of a stochastic matrix, the physical intuitions discussed in the classical case apply essentially unchanged. We can now add to those intuitions the standard example of a sequence of $n$ laser pulses, each involving a fixed energy submanifold. Naturally, these will count as $n$ time-steps. In analogy with the classical counterpart, the quantum time cost admits more or less restrictive definitions. Since we eventually want to prove a quantum advantage, it is convenient to adopt in the quantum regime the simpler and more constrained definition. Thus, we we will use the definition involving time-independent generators only, as an advantage proven according to such a restricted scenario will persist if we allow even more freedom in the quantum protocol.

\begin{defn}[Quantum time cost]
	\label{def:quantum_time}
	The \emph{quantum time cost} $Q_{\rm time}(P,m)$ of a $d\times d$ stochastic matrix $P$, while allowing for $m$ memory states, is the minimum $\tau$ such that there exist time-independent Lindbladians $\mathcal{L}_1, \dots \mathcal{L}_\tau$ on a $(d+m)\times(d+m)$ dimensional Hilbert space and
	\begin{equation}
		Q_{i|j} = \bra{i}e^{\mathcal{L}^{(\tau)} t_\tau} \cdots e^{\mathcal{L}^{(1)} t_1}(\ketbra{j}{j})\ket{i}
	\end{equation}
	implements $P$.
\end{defn}

At this point it is worth highlighting that this notion of time-step, which is a quantum version of that in Ref.~\cite{wolpert2019space}, is disconnected from the one used in quantum computing. There, one takes the system to consist of a number of qubits and the operations are constrained by the fact that they involve only a few qubits at once. Here, instead, we allow for transformations involving all energy levels at once. The basic disconnect is due to the fact that elementary operations in a computational setting typically do not involve dissipative dynamics, which is however the main subject of discussion here (for a unitary it is not natural to act on all energy levels at once, but for a thermalisation process with a collective bath it can be). Hence, it is not obvious to what extent one can reconcile the two approaches. In this work we will adopt ours as working definitions which carry their own physical intuitions and have the advantage of allowing a clear comparison with the classical work in Ref.~\cite{wolpert2019space}. Nonetheless, we believe an approach which captures more directly computational restrictions in the simulation of stochastic processing using both elementary gates and dissipative interactions would be extremely interesting.

Moving on, the central question in the classical setting is to find $C_{\rm space}(P)$ and then characterise $C_{\rm time}(P,m)$ for \mbox{$m\geq C_{\rm space}(P)$}. The main result of Ref.~\cite{wolpert2019space} was to solve this problem for stochastic matrices $P$ that are \mbox{$\{0,1\}$-valued} or, in other words, represent a function $f$ over the set of states $\{1,\dots,d\}$. How do these results compare with what can be done quantum mechanically? In the next section we give a protocol realising every \mbox{$\{0,1\}$-valued} stochastic matrix that scales much better than the corresponding minimal classical cost.

%------------------------------------------------------------------

\subsection{Quantum advantage} 

Let $P_f$ be a $\{0,1\}$-valued $d\times d$ stochastic matrix defined by a function $f:~\mathbb{Z}_d \rightarrow \mathbb{Z}_d$. Let ${\rm fix}(f)$ be the number of fixed points of $f$, $|{\rm img}(f)|$ the dimension of the image of $f$ and ${\rm c}(f)$ the number of cycles of $f$, i.e. the number of distinct orbits of elements of $\{1,\dots,d\}$ of the form $\{i,f(i),f (f(i)), \dots, i\}$. Recently the following result has been shown.

\begin{thm}[Classical cost of a function~\cite{wolpert2019space}]
	\label{thm:trade-off}
	The time cost of a $\{0,1\}$-valued stochastic matrix $P_f$ described by a function $f$ is given by
	\begin{align}
		&C_{\rm time}(P_f,m)\nonumber\\
		 &\quad = \left\lceil \frac{m + d + {\rm max}[{\rm c} (f) - m,0] - {\rm fix}(f)}{m + d - |{\rm img}(f)|} \right\rceil + b_f(m) \nonumber \\
	& \quad\geq \left\lceil \frac{m + d  - {\rm fix}(f)}{m + d - |{\rm img}(f)|} \right\rceil,
	\end{align}
	where $b_f(m) = 0$ or $1$ and $\lceil \cdot \rceil$ is the ceiling function.
\end{thm}

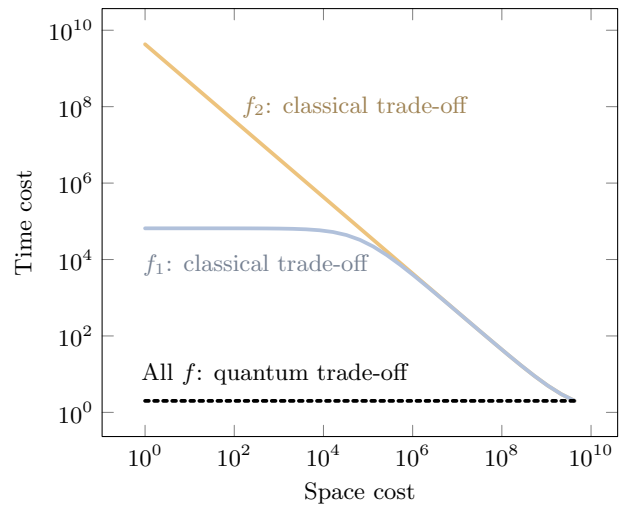
\begin{figure}[t]
	\centering
	\begin{tikzpicture}[line cap=round,line join=round,x=6.0cm,y=6.0cm]
		% Define axes and ticks
		\begin{loglogaxis}[xlabel=Space cost,ylabel=Time cost]
		\addplot[color=tikzOrange,line width=0.5mm] coordinates {
			(1., 4.29497*10^9) (2., 2.14748*10^9) (4., 1.07374*10^9) (8., 5.36871*10^8) (16., 
			2.68435*10^8) (32., 1.34218*10^8) (64., 6.71089*10^7) (128., 
			3.35544*10^7) (256., 1.67772*10^7) (512., 8.38861*10^6) (1024., 
			4.19431*10^6) (2048., 2.09715*10^6) (4096., 1.04858*10^6) (8192.,
			524289.) (16384., 262145.) (32768., 131073.) (65536., 
			65537.) (131072., 32769.) (262144., 16385.) (524288., 
			8193.) (1.04858*10^6, 4097.) (2.09715*10^6, 2049.) (4.1943*10^6, 
			1025.) (8.38861*10^6, 513.) (1.67772*10^7, 257.) (3.35544*10^7, 
			129.) (6.71089*10^7, 65.) (1.34218*10^8, 33.) (2.68435*10^8, 
			17.) (5.36871*10^8, 9.) (1.07374*10^9, 5.) (2.14748*10^9, 
			3.) (4.29497*10^9, 2.)
		};
		\addplot[color=tikzBlue,line width=0.5mm] coordinates {	
			(1., 65535.) (2., 65534.) (4., 65532.) (8., 65528.) (16., 65520.) (32., 
			65504.) (64., 65472.1) (128., 65408.3) (256., 65281.) (512., 
			65028.) (1024., 64527.8) (2048., 63550.1) (4096., 
			61681.) (8192., 58254.3) (16384., 52429.) (32768., 
			43691.) (65536., 32768.5) (131072., 21846.) (262144., 
			13108.) (524288., 7282.67) (1.04858*10^6, 3856.) (2.09715*10^6, 
			1986.91) (4.1943*10^6, 1009.23) (8.38861*10^6, 
			509.023) (1.67772*10^7, 256.) (3.35544*10^7, 
			128.749) (6.71089*10^7, 64.9366) (1.34218*10^8, 
			32.9839) (2.68435*10^8, 16.9959) (5.36871*10^8, 
			8.9989) (1.07374*10^9, 4.99969) (2.14748*10^9, 
			2.99991) (4.29497*10^9, 1.99997)
		};
		\addplot[color=black,line width=0.5mm,dotted] coordinates {	
			(1, 2)  (4.29497*10^9, 2)
		};
		% Labels
		\node at (800,10) {\color{black} All $f$: quantum trade-off};
		\node at (300,7000) {\color{tikzBlue2} $f_1$: classical trade-off};
		\node at (50000,100000000) {\color{tikzOrange2} $f_2$: classical trade-off};
		\end{loglogaxis}
	\end{tikzpicture}
	\caption{\label{fig:trade-off}\textbf{Classical versus quantum space-time trade-off.} The optimal trade-off between space cost and time cost of implementing stochastic matrices for a system of $s=32$ bits, i.e., with dimension $d=2^{32}$ (plotted in log-log scale). Solid coloured curves correspond to optimal trade-offs for classically implementing exemplary $\{0,1\}$-valued stochastic matrices described by functions \mbox{$f_1(i)=i\oplus 1$} (addition modulo $d$) and \mbox{$f_2(i)=\min\{i+2^{s/2},2^s-1\}$}, as analysed in Ref.~\cite{wolpert2019space}. Dashed black curve corresponds to optimal trade-offs for quantumly implementing any $\{0,1\}$-valued stochastic matrix, thus illustrating a quantum advantage.}
\end{figure}

Suppose that the state space is given by all bit strings of length $s$, so that $d = 2^s$. Theorem~\ref{thm:trade-off} shows that, if $|{\rm img}(f)|$ is $O(d)$, then $P_f$ is expensive to simulate by memoryless dynamics unless the number of fixed points is also $O(d)$. Since for a typical $f$ we have $|{\rm img}(f)| = O(d)$ and ${\rm fix}(f) = O(1)$ (see Appendix~\ref{app:typical}), we conclude that typically \mbox{$C_{\rm time}(P_f,m) = O(2^s/m)$}, i.e., an exponential number of memory states are required to have an efficient simulation in the number of time-steps. Conversely, one needs an exponential number of time-steps to have an efficient simulation for a fixed number of memory states. One of the examples discussed in Ref.~\cite{wolpert2019space} is that of \mbox{$f_1(i) = i \oplus 1$} (addition modulo $d$), which may be interpreted as keeping track of a clock in a digital computer. From Theorem~\ref{thm:trade-off} we see that one has \mbox{$C_{\rm time}(P_{f_1},m) \geq 2^s/m$} with $m$ the number of memory states introduced~(see Fig.~\ref{fig:trade-off}). However, as we discussed already above, any permutation is quantum embeddable by a unitary, and hence $Q_{\rm time}(P_{f_1},0) =1$. The existence of this advantage is generalised by the following result.

\begin{thm}[Quantum cost of a function]
	\label{thm:q_trade-off}
	For any \mbox{$m \geq 0$} and any function $f$ we have $Q_{\rm time}(P_f,m) \leq 2$.
\end{thm}

The explicit proof is given in Appendix~\ref{app:trade_off}, but it is based on the simple fact that every function can be realised quantumly by a unitary process realising a permutation followed by a classical master equation achieving an idempotent function $f_I$ in a single time-step. Hence, one can achieve every function quantumly using zero memory states and only two time-steps, compared to the typical classical cost \mbox{$C_{\rm time}(P_{f_1},m) \geq d/m$}. This result, illustrated in Fig.~\ref{fig:trade-off}, is quantitative evidence of the power of superposition to act as effective memory.

%------------------------------------------------------------------

\section{Role of memory in state transformations}
\label{sec:state}

\subsection{Accessibility regions}

In this section, we change the focus from processes to states. We will investigate whether a given state transformation can be realised by either a classical or a quantum master equation. In other words, given an input distribution $\v{p}$, is it possible to get a given final state $\v{q}$ through an embeddable (or quantum embeddable) stochastic matrix~$P$? Of course, given full control it is always possible to choose a master equation with $\v{q}$ being the unique fixed point. However, more realistically, a fixed point of the evolution is constrained rather than being arbitrary -- and typically corresponds to the thermal Gibbs distribution $\v{\gamma}$ with
\begin{equation}
	\label{eq:thermalstate}
	\gamma_k:=\frac{1}{Z} e^{-\beta E_k},\quad Z:=\sum_{k=1}^d e^{-\beta E_k}.
\end{equation}
Here, $E_k$ are the energy levels of the system  interacting with an external environment at inverse temperature $\beta$. Hence, suppose some (full-rank) fixed point $\v{\gamma}$ is given (which may or may not be the thermal state of the system). We then introduce the following two definitions and discuss the corresponding examples involving standard dynamical models. The former encapsulates the set of input-output relations achievable by means of general stochastic processes with a fixed point $\v{\gamma}$, while the latter captures the subset achievable without exploiting memory effects, i.e. by Markovian master equations. 

\begin{defn}[Classical accessibility]
	\label{def:reach_class}
	A distribution $\v{q}$ is accessible from $\v{p}$ by a classical stochastic process with a fixed point $\v{\gamma}$ if there exists a stochastic matrix $P$, such that
	\begin{equation}
	P\v{p} = \v{q} \quad \textrm{and} \quad P\v{\gamma} = \v{\gamma}.
	\end{equation}
We denote the set of all $\v{q}$ accessible from $\v{p}$ given $\v{\gamma}$ by $\mathcal{C}^{\rm Mem}_{\v{\gamma}}(\v{p})$.
\end{defn}

\begin{ex}
	\label{ex:JC}
	A standard effective model in cavity QED and atomic physics is the Jaynes-Cummings model in the rotating wave approximation~\cite{breuer2002open}. Formally it describes the resonant interaction of a two-level system (with energy levels $\ket{0}$ and $\ket{1}$) with a single harmonic oscillator by means of the Hamiltonian:
	\begin{equation}
		H_{JC} = \frac{\hbar \omega}{2} \sigma_z + \hbar \omega a^\dag a + g(t) (\sigma^+ a + \sigma^- a^\dag), 
	\end{equation}
	where $a^\dag$ and $a$ are the creation and annihilation operators of the oscillator, and \mbox{$\sigma_z=\ketbra{0}{0}-\ketbra{1}{1}$}, \mbox{$\sigma^+ = \ketbra{1}{0}$}, \mbox{$\sigma^- = \ketbra{0}{1}$}. Suppose the oscillator is initially in a thermal state, the system is in a general state $\rho(0)$, and denote by $\v{p}(t)$ the populations in the energy eigenbasis of the system at time $t$. Then, in the reduced dynamics of the system generated by $H_{JC}$, the population and coherence terms decouple. Moreover, one can show that \mbox{$\v{p}(t) = P_t \v{p}(0)$}, where \mbox{$P_t \v{\gamma} = \v{\gamma}$} and $\v{\gamma}$ is the thermal distribution of the two-level system. Hence, \mbox{$\v{p}(t)\in \mathcal{C}^{\rm Mem}_{\v{\gamma}}(\v{p}(0))$}, i.e. if we know $\mathcal{C}^{\rm Mem}_{\v{\gamma}}(\v{p}(0))$ we can constrain the set of achievable final states. In fact, in the low temperature regime, $\mathcal{C}^{\rm Mem}_{\v{\gamma}}(\v{p}(0))$ is a good approximation of the set of all states that can be achieved in the Jaynes-Cummings model after a long enough dynamics~\cite{lostaglio2018elementary}.  
\end{ex}

In what follows we denote by $\rho_{\v{p}}$ the density matrix diagonal in the computational (or energy) basis with entries given by the probability distribution $\v{p}$: \mbox{$\rho_{\v{p}} = \sum_k p_k \ketbra{k}{k}$}.

\begin{ex}
	Consider a system in a classical state $\rho_{\v{p}}$. Then, $\mathcal{C}^{\rm Mem}_{\v{\gamma}}(\v{p})$ describes the set of classical states that can be obtained from $\rho_{\v{p}}$ by \emph{thermal operations}~\cite{horodecki2013fundamental,lostaglio2019introductory}, i.e., energy-preserving couplings of the system with arbitrary thermal baths at a temperature fixed by the choice of $\v{\gamma}$. This scenario was analysed in Refs.~\cite{ruch1978mixing,horodecki2013fundamental}, and it was proven there that $\mathcal{C}^{\rm Mem}_{\v{\gamma}}(\v{p})$ is fully specified by the notion of thermo-majorisation (also known as \emph{majorisation relative to $\v{\gamma}$}~\cite{marshall2010inequalities}). 
\end{ex}

The memoryless version of the above classical accessibility region is defined as follows.

\begin{defn}[Classical memoryless accessibility]
	\label{def:reach_class_mem}
	A distribution $\v{q}$ is accessible from $\v{p}$ by a classical master equation with a fixed point $\v{\gamma}$ if there exists a continuous one-parameter family $L(t)$ of rate matrices generating a family of stochastic matrices $P(t)$, such that
	\begin{equation}
	P(t_f)\v{p} = \v{q}, \quad L(t)\v{\gamma} = \mathbf{0} \quad \textrm{for all } t \in [0,t_f)
	\end{equation} We denote the set of all $\v{q}$ accessible from $\v{p}$ given $\v{\gamma}$ by $\mathcal{C}_{\v{\gamma}}(\v{p})$.
\end{defn}

\begin{ex}
	Consider an $N$-level quantum system weakly interacting with a large thermal bath. A microscopic derivation~\cite{davies1974markovian} leads to a quantum Markovian master equation for the system known as a Davies process~\cite{roga2010davies}, which is a standard model for thermalisation~\cite{alicki2018introduction}. One can then verify that the populations $\v{p}(t)$ of the system undergoing a Davies process satisfy a classical Markovian master equation,
	\begin{equation}
		\frac{d\v{p}(t)}{dt} = L \v{p}(t), \quad L \v{\gamma} = 0,
	\end{equation}
	and thus \mbox{$\v{p}(t) \in  \mathcal{C}_{\v{\gamma}}(\v{p}(0))$}. A characterisation of $\mathcal{C}_{\v{\gamma}}(\v{p})$ then allows one to restrict the intermediate non-equilibrium states generated by standard thermalisation processes, without the need to solve the actual dynamics. Since these processes are building blocks of more complex protocols -- they form thermalisation strokes of engines and refrigerators~\cite{uzdin2015equivalence} -- this information can be employed to identify thermodynamic protocols with optimal performance.
\end{ex}

These definitions naturally generalise to quantum dynamics in the following way.
\begin{defn}[Quantum accessibility]
	\label{def:reach_quant}
	A distribution $\v{q}$ is accessible from $\v{p}$ by quantum dynamics with a fixed point $\v{\gamma}$ if there exists a quantum channel $\E$, such that 
	\begin{equation}
	\E(\rho_{\v{p}})= \rho_{\v{q}} \quad  \textrm{and} \quad \E(\rho_{\v{\gamma}}) = \rho_{\v{\gamma}}.
	\end{equation}
We denote the set of all $\v{q}$ accessible from $\v{p}$ given $\v{\gamma}$ by $\mathcal{Q}^{\mathrm{Mem}}_{\v{\gamma}}(\v{p})$. 
\end{defn}

\begin{ex} 	
	Setting $\v{\gamma}$ to be the thermal Gibbs state, the set of channels satisfying $\mathcal{E}(\rho_{\v{\gamma}}) = \rho_{\v{\gamma}}$ (\emph{Gibbs-preserving maps}) was identified as the most general set of operations that can be performed without investing work~\cite{faist2018fundamental}. The reason is that any channel that is not Gibbs-preserving can create a non-equilibrium resource from an equilibrium state. These channels can then be taken as free operations and the minimal work cost of a general channel $\mathcal{F}$ can be computed from the minimal work a battery system must provide to simulate $\mathcal{F}$ using Gibbs-preserving maps only. The set of stochastic matrices $P$ that can be realized with Gibbs-preserving maps coincides with the set of all $P$ with a fixed point $\v{\gamma}$ (see, e.g., Theorem~1 of Ref.~\cite{lostaglio2019introductory}). It is not surprising that this includes $P$ that cannot be simulated without memory effects, since Gibbs-preserving quantum channels in general cannot be realized by means of a Markovian master equation, i.e. they require memory effects.
\end{ex} 

It is of course natural to consider the transformations that can be realised by the subset of Gibbs-preserving maps that originate from a Markovian quantum master equation, and so we introduce the following.

\begin{defn}[Quantum memoryless accessibility]
	\label{def:reach_quant_mem}
	A distribution $\v{q}$ is accessible from $\v{p}$ by a quantum master equation with a fixed point $\v{\gamma}$ if there exists a continuous one-parameter family of Lindbladians $\mathcal{L}(t)$ generating a family of quantum channels $\mathcal{E}(t)$, such that
	\begin{equation}
	\mathcal{E}(t_f)[\rho_{\v{p}} ]= \rho_{\v{q}}, \quad  \mathcal{L}(t)[\rho_{\v{\gamma}}] = 0 \quad \textrm{for all } t \in [0,t_f).
	\end{equation} 
We denote the set of all $\v{q}$ accessible from $\v{p}$ given $\v{\gamma}$ by $\mathcal{Q}_{\v{\gamma}}(\v{p})$. 
\end{defn}

\begin{ex} 
	Standard thermalisation processes resulting from weak couplings to a large environment, such as Davies maps, are generated by a Lindbladian $\mathcal{L}$ satisfying  $\mathcal{L}[\rho_{\v{\gamma}}] = 0$, as required by the definition above. However, these dynamics are unable to create quantum superpositions and hence cannot be used to show a quantum advantage. On the other hand, more exotic thermalisation processes exist. Let $\ket{\gamma}:= \sqrt{\gamma_0} \ket{0} + \sqrt{\gamma_1} \ket{1}$, $\gamma = \gamma_0 \ketbra{0}{0} + \gamma_1 \ketbra{1}{1}$ and consider the quantum master equation on a two-level system
	\begin{equation}
		\frac{d}{dt} \rho = \mathcal{L}(\rho), 
	\end{equation} 
	where $\L$ is the Lindbladian specified by Eq.~\eqref{eq:lindbladian} with a vanishing Hamiltonian $H$ and the map $\Phi$ given by a measure-and-prepare channel of the following form: 
	\begin{equation}
		\Phi(\rho) = \frac{\rho_{00}}{\gamma_0}(\gamma - \gamma_1 \ketbra{\gamma}{\gamma}) + \rho_{11} \ketbra{\gamma}{\gamma}.
	\end{equation}
	It satisfies $\mathcal{L}(\rho_{\gamma}) = 0$, so if $\v{p}(t)$ is the population vector of $\rho(t)$, one has \mbox{$\v{p}(t) \in \mathcal{Q}_{\v{\gamma}}(\v{p}(0))$.} One can verify that the dynamics equilibrates every state to $\gamma$ and yet it is capable of generating coherence (for example, for $\rho(0) = \ketbra{1}{1}$ one has $d \rho_{01}/dt  > 0$ around $t = 0$). This particular dynamics is unable to translate the ability to create coherence into the ability to generate a $P$ that classically requires memory (we will later construct dynamics that do). What the example illustrates, however, is a central mechanism by which a quantum advantage can arise, i.e. generation of quantum superpositions by exotic thermalising Markovian master equations. These dynamics are in principle allowed by quantum mechanics, but we leave open for future works how they can be actually realized. We note, in passing, that what is needed is a physical model where a Markovian process with a thermal fixed point naturally emerges despite the fact that the dynamics does not satisfy the secular approximation.
\end{ex}

Note that in Definitions~\ref{def:reach_class_mem}~and~\ref{def:reach_quant_mem}, the requirements \mbox{$L(t) \v{\gamma} = \v{0}$} and \mbox{ $\mathcal{L}(t) [\rho_{\v{\gamma}}] = 0$} for all times ensure that \mbox{$P(t)\v{\gamma} = \v{\gamma}$} and \mbox{$\mathcal{E}(t)[\rho_{\v{\gamma}}] = \rho_{\v{\gamma}}$} for all intermediate times $t \in [0, t_f]$. By characterising the difference between the sets $\mathcal{C}^{\rm Mem}_{\v{\gamma}}(\v{p})$ and $\mathcal{C}_{\v{\gamma}}(\v{p})$, one can capture the state transformations that can be achieved only through controls exploiting memory effects. That is, all states $\v{q}\in \mathcal{C}^{\rm Mem}_{\v{\gamma}}(\v{p})$, but not in $\mathcal{C}_{\v{\gamma}}(\v{p})$, can only be achieved from $\v{p}$ via a transformation that employs memory. Analogous statements hold for $\mathcal{Q}^{\rm Mem}_{\v{\gamma}}(\v{p})$ and $\mathcal{Q}_{\v{\gamma}}(\v{p})$.

In what follows we will study relations between the accessibility regions. Our main result will be that \mbox{$\mathcal{C}_{\v{\gamma}}(\v{p}) \subset \mathcal{Q}^{}_{\v{\gamma}}(\v{p})$} in the case of a uniform fixed point $\v{\eta}:=(1/d,\dots,1/d)$ and for \emph{general} fixed points for a qubit system. These results signal a quantum advantage, i.e., some transitions that classically require memory can be achieved through memoryless  quantum dynamics. 

%------------------------------------------------------------------

\subsection{Quantum advantage at infinite temperature}
\label{sec:qubit_ext}
 
A natural question that arises is whether the sets $\mathcal{Q}_{\v{\gamma}}(\v{p})$ and $\mathcal{Q}^{\rm Mem}_{\v{\gamma}}(\v{p})$ are larger than their classical counterparts. This enlargement of the set of achievable states is another facet of the quantum advantage. More generally, one could also investigate more refined versions of memory advantages on state transformations, e.g. trying to include space and time resources in the analysis, similarly as we did it in Sec.~\ref{sec:space_time}.

It is straightforward to prove that without the memoryless constraint there will be no quantum advantage. In other words, we have $\v{q}\in\mathcal{Q}^{\rm Mem}_{\v{\gamma}}(\v{p})$ if and only if $\v{q}\in \mathcal{C}^{\rm Mem}_{\v{\gamma}}(\v{p})$. The ``if'' part is obvious, as the set of all quantum channels with a fixed point $\rho_{\v{\gamma}}$ contains as a subset the set of classical stochastic processes with the same fixed point. Conversely, take any $\v{q}\in\mathcal{Q}^{\rm Mem}_{\v{\gamma}}(\v{p})$, meaning that there exists a channel $\E$ such that $\E(\rho_{\v{p}})=\rho_{\v{q}}$ and $\E(\rho_{\v{\gamma}})=\rho_{\v{\gamma}}$. Then we can construct a stochastic process $P$ with matrix elements $P_{i|j}$ given by $\matrixel{i}{\mathcal{E}\left(\ketbra{j}{j}\right)}{i}$. Matrix $P$ is stochastic because $\mathcal{E}$ is positive and trace-preserving. Furthermore, it satisfies $P\v{p}=\v{q}$ and $P\v{\gamma}=\v{\gamma}$. Therefore, $\v{q}\in\mathcal{C}^{\rm Mem}_{\v{\gamma}}(\v{p})$. To sum up:
\begin{equation}
	\mathcal{Q}^{\rm Mem}_{\v{\gamma}}(\v{p}) = 	\mathcal{C}^{\rm Mem}_{\v{\gamma}}(\v{p}).
\end{equation}

However, as we will now prove, a quantum advantage is exhibited by $\mathcal{Q}_{\v{\gamma}}(\v{p}) \supsetneq \mathcal{C}_{\v{\gamma}}(\v{p})$, i.e., there are states classically accessible only with memory that can be achieved by quantum memoryless dynamics. In the case of a uniform fixed point, going from classical to quantum memoryless dynamics allows one to achieve the maximal quantum advantage: all transformations involving memory can be realised quantum mechanically with no memory.

\begin{thm}[Maximal quantum advantage for uniform fixed points]
	\label{thm:advantageinfinitetemperature} For every $\v{p}$ and a uniform distribution \mbox{$\v{\eta} = (1/d, \dots, 1/d)$} one has $\mathcal{Q}_{\v{\eta}}(\v{p}) = \mathcal{C}^{\rm Mem}_{\v{\eta}}(\v{p})$. 
\end{thm}
\begin{proof}
	Clearly we have $\mathcal{Q}_{\v{\eta}}(\v{p}) \subseteq \mathcal{C}^{\rm Mem}_{\v{\eta}}(\v{p})$, because \mbox{$\mathcal{Q}_{\v{\eta}}(\v{p}) \subseteq \mathcal{Q}^{\rm Mem}_{\v{\eta}}(\v{p}) = \mathcal{C}^{\rm Mem}_{\v{\eta}}(\v{p})$}. Conversely, let us show \mbox{$ \mathcal{C}^{\rm Mem}_{\v{\eta}}(\v{p})\subseteq  \mathcal{Q}_{\v{\eta}}(\v{p})$.}  Take $\v{q}\in \mathcal{C}^{\rm Mem}_{\v{\eta}}(\v{p})$. As it is well-known (see Theorem B6 of \cite{marshall2010inequalities}) this is equivalent to majorisation relation $\v{p}\succ\v{q}$, i.e.,
	\begin{equation}
		\sum_{i=1}^k p^{\downarrow}_i \geq \sum_{i=1}^k q^{\downarrow}_i, \quad k=1, \dots d,
	\end{equation}	
	where $\v{p}^\downarrow$ denotes the probability distribution $\v{p}$ sorted in a non-increasing order. We will show that every $\v{q}$ satisfying the above can be achieved from $\v{p}$ by a composition of two quantum embeddable processes (so, according to Lemma~\ref{lem:composition}, by a quantum-embeddable process), each with a uniform fixed point. First, note that every permutation is quantum-embeddable, as discussed in Sec.~\ref{sec:examples}. Thus, one can rearrange $\v{p}$ into $\v{p}'$ with $\v{p}' \succ \v{p}$ and sorted in the same way as $\v{q}$. By transitivity of majorisation we have \mbox{$\v{p'}\succ\v{q}$}. Now, using Theorem~11 of Ref.~\cite{perry2018sufficient} with $\beta \rightarrow \infty$, it follows that $\v{q}$ can be achieved from $\v{p'}$ by applying a sequence of stochastic processes of the form
	\begin{equation}
		T^{(i,j)}:= \begin{bmatrix}
			1- \lambda/2 & \lambda/2 \\\lambda/2 & 1-\lambda /2
		\end{bmatrix} \oplus I_{\setminus_{(i,j)}},
	\end{equation}
	where $\lambda \in [0,1]$ and $I_{\backslash_{(i,j)}}$ is the identity matrix on all states excluding $i$ and $j$. Matrices $T^{(i,j)}$ are embeddable with rate matrices 
	\begin{equation}
		L^{(i,j)}= \begin{bmatrix}
		-1/2 & 1/2 \\ 1/2 & -1/2
		\end{bmatrix} \oplus \v{0}_{\backslash_{(i,j)}}
	\end{equation}
	where $\v{0}_{\backslash_{(i,j)}}$ is the zero matrix on all states excluding $i$ and $j$. Moreover, matrices $L^{(i,j)} $ satisfy $L^{(i,j)}\v{\eta} = \v{0}$. Putting everything together, we have a unitary permutation followed by a sequence of processes $e^{L^{(i_k,j_k)} t_k}$ that map $\v{p}$ into $\v{q}$, so $\v{q} \in \mathcal{Q}_{\v{\eta}}(\v{p})$. We conclude that $\mathcal{Q}_{\v{\eta}}(\v{p}) = \mathcal{C}^{\rm Mem}_{\v{\eta}}(\v{p})$.
\end{proof}

Since $\mathcal{Q}_{\v{\eta}}(\v{p}) = \mathcal{C}^{\rm Mem}_{\v{\eta}}(\v{p})$, in order to prove quantum advantage we only need to show that classical memoryless dynamics is more restrictive than general classical dynamics with a fixed point, i.e., that $\mathcal{C}_{\v{\eta}}(\v{p})$ is a proper subset of $\mathcal{C}^{\rm Mem}_{\v{\eta}}(\v{p})$. This is indeed the case, as can be easily verified for $d=2$ (and it is rigorously proven in Ref.~\cite{lostaglio2020fundamental} for a general fixed points $\v{\gamma}$ in any dimension). Here, in Fig.~\ref{fig:uni_advantage}, we illustrate that $\mathcal{C}_{\v{\eta}}(\v{p})$ is a proper subset of $\mathcal{C}^{\rm Mem}_{\v{\eta}}(\v{p})$ for an exemplary case of $d=3$.

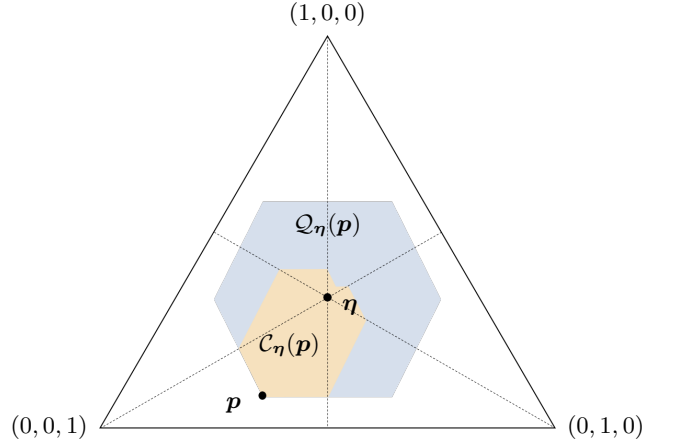
\begin{figure}[t]
	\centering
	\begin{tikzpicture}
	% External picture
	\node (myfirstpic) at (0,0) {\includegraphics[width=0.7\columnwidth]{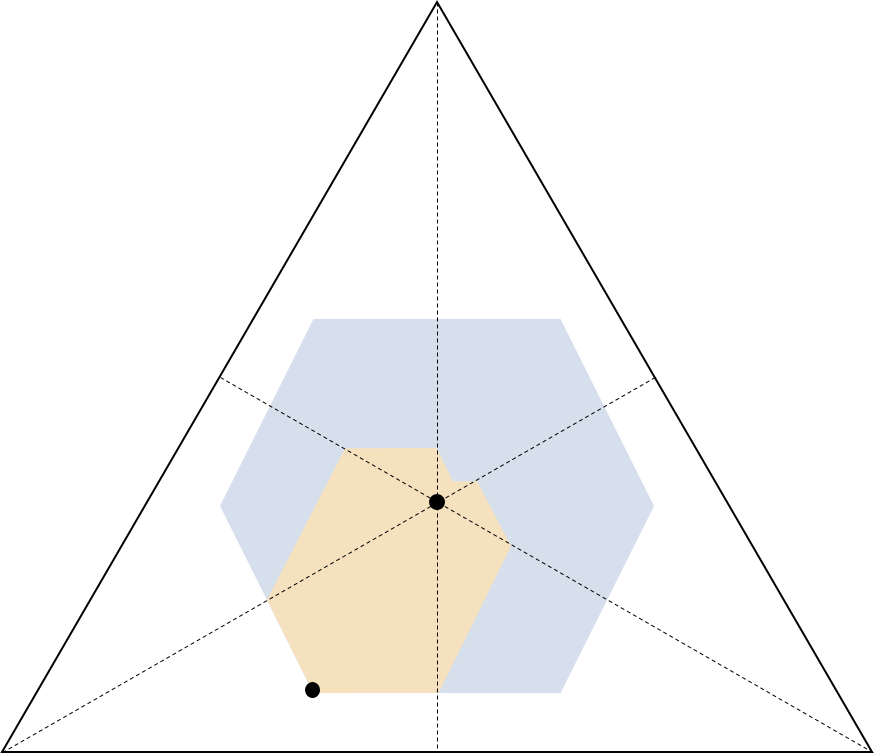}};
	% Labels
	\node at (0,2.9) {\color{black} $(1,0,0)$};
	\node at (3.7,-2.6) {\color{black} $(0,1,0)$};
	\node at (-3.7,-2.6) {\color{black} $(0,0,1)$};
	\node at (-1.25,-2.3) {\color{black} $\v{p}$};
	\node at (0.3,-1) {\color{black} $\v{\eta}$};
	\node at (0,0.1) {$\mathcal{Q}_{\v{\eta}}(\v{p})$};
	\node at (-0.5,-1.5) {$\mathcal{C}_{\v{\eta}}(\v{p})$};
	\end{tikzpicture}
	\caption{\label{fig:uni_advantage} \textbf{Quantum advantage at infinite temperature for $d=3$.} The sets of states accessible via classical ($\mathcal{C}_{\v{\eta}}(\v{p})$, smaller orange shape) and quantum ($\mathcal{Q}_{\v{\eta}}(\v{p})$, larger blue hexagon) memoryless dynamics with a uniform fixed point $\v{\eta}$ for a system of dimension $d=3$ and an exemplary initial state $\v{p}$ (each point inside a triangle corresponds to a probabilistic mixture of sharp distributions).}
\end{figure}

\subsection{Quantum advantage at any finite temperature (qubit case)}

Now we generalize the considerations above to arbitrary fixed points or, if we think of $\v{\gamma}$ as a thermal fixed point, arbitrary temperatures. In other words, we compare states achievable by classical memoryless processes with a fixed point $\v{\gamma}$, to states achievable by quantum memoryless processes with a fixed point $\rho_{\v{\gamma}}$. Here, we will solve the case for $d=2$. To ease physical intuition, we will parametrize $\v{\gamma}$ as in Eq.~\eqref{eq:thermalstate}, with $E=E_2-E_1$ the energy gap between the two states and $\beta$ the inverse temperature of the external environment. 

Let us start by recalling the classical solution to this problem. First, $\mathcal{C}^{\rm Mem}_{\v{\eta}}(\v{p})$ can be obtained using the thermo-majorisation condition \cite{horodecki2013fundamental,  ruch1978mixing}. If $\v{p} = (p,1-p)$, in terms of the achievable ground state populations we have
\begin{equation}
	\label{eq:achievablewithmemory}
\mathcal{C}^{\rm Mem}_{\v{\gamma}}(\v{p}) =  
\begin{cases}
		[p, 1-e^{-\beta E } p] \quad \textrm{if} \quad p\leq{1}/{Z}, \\ 
		[1-e^{-\beta E } p,p] \quad \textrm{if} \quad p> {1}/{Z} .
	\end{cases}
\end{equation}
This set of states can be approximately achieved by the Jaynes-Cummings interaction of Example~\ref{ex:JC}, which involves memory effects. The need for memory effects to achieve the states in Eq.~\eqref{eq:achievablewithmemory} can be understood more generally by looking at information backflows, which are a standard signature of non-Markovian effects. As one can verify from the previous expressions, any continuous trajectory $\v{r}(t)$ connecting $\v{p}$ to the other extremal point in $\mathcal{C}^{\rm Mem}_{\v{\gamma}}(p)$ has to pass through $\v{\gamma}$. This means that the (non-equilibrium) free energy of the system,
\begin{equation}
	\label{eq:free}
	\!\!	F(\v{p})= \sum_i p_i E_i - \frac{1}{\beta}  H(\v{p}),~~ H(\v{p}) = - \sum_i p_i \log p_i,\!
\end{equation}
will first decrease all the way to its minimal value, and then increase again, thus signalling an information back-flow from the thermal environment. It is obvious that such a phenomenon does not occur when dissipation is well-described by a Markovian master equation, since in this case the dynamics cannot cross the fixed state~$\v{\gamma}$. Memory effects (e.g., in the form of non-negligible system-bath correlations) are required to access all the states in $\mathcal{C}^{\rm Mem}_{\v{\gamma}}(p)$.

Let us now look at $\mathcal{C}_{\v{\gamma}}(p)$, the achievable states with classical memoryless processes with a thermal fixed point. For $d=2$, the set of Markovian master equations with a fixed point $\v{\gamma}$ is limited to the thermalisations introduced in Example~\ref{ex:classicalthermalization}. Hence, $\mathcal{C}_{\v{\gamma}}(p)$ is readily characterised as all states along a line connecting $\v{p}$ and the thermal state:
\begin{equation}
	\label{eq:achievablenomemory}
	\mathcal{C}_{\v{\gamma}}(\v{p}) =  
	\begin{cases}
		\left[p,1/Z\right] \quad \textrm{if} \quad p\leq 1/Z, \\ 
		\left[1/Z,p\right] \quad \textrm{if} \quad p> 1/Z. 
	\end{cases}
\end{equation}
The difference between Eqs.~\eqref{eq:achievablewithmemory}~and~\eqref{eq:achievablenomemory} is simply that memory effects allow one to go ``on the other side'' of the thermal state. Without memory this is not possible, since the thermal state is at a minimum of the free energy.

Now we turn to the corresponding quantum mechanical problem, looking for potential advantages. Any unitary that changes the population of the state is forbidden since it does not have $\rho_{\v{\gamma}}$ as a fixed point. We cannot hence rely on any of the previous constructions exploiting the fact that quantum mechanically one can generate permutations without using memory. We then need to characterise the set of diagonal quantum states achievable from a given state $\rho_{\v{p}}$ via Markovian quantum master equations with a given fixed point $\rho_{\v{\gamma}}$. Even without the constraint that the channel is generated by a master equation, finding a simple characterisation of the set of accessible states for $d>2$ has remained an open problem for decades~\cite{alberti1980problem}. This explains why we are focusing here on the simplest non-trivial case of a qubit system, where such problem has been fully solved~\cite{alberti1980problem,gour2017quantum, korzekwa2017structure}. We will numerically show that in $d=2$ one also achieves a maximal quantum advantage. Specifically, all states achievable classically by means of processes with thermal fixed points that involve memory, Eq.~\eqref{eq:achievablewithmemory}, can be attained by a Markovian quantum master equation with fixed point $\rho_{\v{\gamma}}$. More compactly, we have the following result:

\begin{res}[Quantum advantage at every finite temperature -- numerics]\label{res:1}
	For $d=2$, $\mathcal{Q}_{\v{\gamma}}(\v{p}) = C^{\rm Mem}_{\v{\gamma}}(\v{p})$.
\end{res}

This showcases that the advantage of Theorem~\ref{thm:advantageinfinitetemperature} is not limited to the special case of a uniform fixed point involving unitary dynamics. Superposition can substitute memory in the control of classical systems at every finite temperature. 

In order to prove Result~\ref{res:1}, we present an explicit construction and numerical evidence for an even stronger result.

\begin{res}[Numerics]
	\label{result2}
Every qubit state accessible via a qubit channel with given fixed point can be achieved by a qubit Markovian master equation with the same fixed point.
\end{res}

Let us start by recalling the result of Ref.~\cite{alberti1980problem}, where the authors provided necessary and sufficient conditions for the existence of a qubit channel $\E$ satisfying:
\begin{equation}
	\E(\rho)=\rho',\quad \E(\sigma)=\sigma',
\end{equation}
for any two pairs of qubit density matrices \mbox{$(\rho,\rho')$} and \mbox{$(\sigma,\sigma')$}. Moreover, whenever such a channel exists, the authors provided a construction of the Kraus operators of~$\E$. Setting \mbox{$\sigma=\sigma'=\rho_{\v{\gamma}}$} one obtains a characterisation of all states accessible from $\rho$ through arbitrary channels with a given fixed point $\rho_{\v{\gamma}}$ (we choose a basis in which the fixed point is diagonal). In Ref.~\cite{korzekwa2017structure} the continuous set of conditions presented in Ref.~\cite{alberti1980problem} was reduced to just two inequalities: 
\begin{equation}
\label{eq:r+-}
	R_\pm(\rho)\geq R_\pm(\rho').
\end{equation}

These are best understood through the standard Bloch sphere parametrisation of the states involved. Recall that a general qubit state can be written as 
\begin{align}
	\rho=\frac{\iden+\v{r}_\rho\cdot \v{\sigma}}{2},
\end{align}
where $\v{\sigma}$ denotes the vector of Pauli matrices \mbox{$(\sigma_x,\sigma_y,\sigma_z)$}, while $\v{r}_\rho$ is a three-dimensional real vector which uniquely represents $\rho$ as a point inside a unit Bloch ball in $\mathbb{R}^3$. We parametrise the initial, final and fixed points as follows:
\begin{equation}
	\label{eq:bloch_param}
\v{r}_\rho=(x,y,z),\quad\v{r}_{\rho'}=(x',y',z'),\quad\v{r}_\gamma=(0,0,\zeta).
\end{equation}
Unitary rotations about the $z$ axis leave $\rho_{\v{\gamma}}$ unchanged. By performing such rotations before and after the channel $\mathcal{E}$, without loss of generality we can set $x\geq 0$, $x'\geq 0$ and $y=y'=0$. The monotones $R_{\pm}$ from Eq.~\eqref{eq:r+-} are then defined as~\cite{korzekwa2017structure}:
\begin{equation}
	R_{\pm}(\rho)=\delta(\rho)\pm\zeta z,
\end{equation}
where
\begin{equation}
	\label{eq:delta_lattice}
	\delta(\rho):=\sqrt{(z-\zeta)^2+x^2(1-\zeta^2)},
\end{equation}
with analogous (primed) definitions for $\rho'$. The two inequalities from Eq.~\eqref{eq:r+-} can be then used to find extremal states accessible from $\rho$ via qubit channels with fixed point $\rho_{\v{\gamma}}$. As shown in Fig.~\ref{fig:qubit_regions}, these are given by:
\begin{itemize}
	\item States with a constant $R_+$ lying on a circle with centre and radius $(\v{c}_0, R_0)$ if $z' \geq z$, where
	\begin{equation}
		R_0=\frac{R_+-\zeta^2}{1-\zeta^2}, \quad \v{c}_0=[0,0,\zeta(1-R_0)].
	\end{equation}		
	\item States with a constant $R_-$ lying on a circle $(\v{c}_1, R_1)$ if $z' < z$, where
	\begin{equation}
		R_1=\frac{R_-+\zeta^2}{1-\zeta^2},\quad \v{c}_1=[0,0,\zeta(1+R_1)].
	\end{equation}
\end{itemize}

\begin{figure}[t]
	\centering
	\begin{tikzpicture}[line cap=round,line join=round,x=2.2cm,y=2.2cm]
		\clip(-1.15130521264221,-1.2346670851205752) rectangle (1.2120740854657621,1.3343488940797);
		\draw [line width=1.2pt,color=tikzOrange,fill=tikzOrange,fill opacity=0.5] (0.,0.12137003521531087) circle (1.1319436901052642cm);
		\draw [line width=1.6pt,color=tikzBlue,fill=tikzBlue,fill opacity=0.5] (0.,0.41196329811802246) circle (1.4252770234385979cm);
		\draw [line width=0.8pt] (0.,0.) circle (2.2cm);
		\draw [line width=1.6pt,color=tikzOrange2] (0.,0.12137003521531087)-- (0.5,0.);
		\draw [line width=1.6pt,color=tikzBlue2] (0.,0.41196329811802246)-- (0.5,0.);
		\draw [line width=0.8pt,dash pattern=on 2pt off 2pt] (0.,1.)-- (0.,-1.);
		\begin{scriptsize}
		\draw [fill=black] (0.5,0.) circle (2.0pt);
		\draw [fill=black] (0.,0.25) circle (2.0pt);
		\draw [fill=tikzOrange2] (0.,0.12137003521531087) circle (2.0pt);
		\draw [fill=tikzBlue2] (0.,0.41196329811802246) circle (2.0pt);
		\draw [fill=black] (0.,1.) circle (2.0pt);
		\draw [fill=black] (0.,-1.) circle (2.0pt);
		\end{scriptsize}
		\node at (0,1.15) {\color{black} $\ketbra{0}{0}$};
		\node at (0,-1.15) {\color{black} $\ketbra{1}{1}$};
		\node at (-0.15,0.25) {\color{black} $\rho_{\v{\gamma}}$};
		\node at (-0.15,0.45) {\color{tikzBlue2} $\v{c}_1$};
		\node at (-0.15,0.1) {\color{tikzOrange2} $\v{c}_0$};
		\node at (0.23,0.37) {\color{tikzBlue2} $R_1$};
		\node at (0.2,-0.05) {\color{tikzOrange2} $R_0$};
		\node at (0.6,0) {\color{black} $\rho$};
	\end{tikzpicture}
	\caption{\label{fig:qubit_regions} \textbf{Qubit accessibility region.}
		Geometrically, states with a fixed value of $R_+$ lie on a circle centred at $\v{c}_0$ and with radius $R_0$ (in orange). Similarly, states with a fixed value of $R_-$ lie on a circle centred at $\v{c}_1$ and with radius $R_1$ (in blue). States achievable from a given initial state $\rho$ via quantum channels with a fixed point $\rho_{\v{\gamma}}$ lie inside the Bloch sphere in the intersection of two balls $(\v{c}_0,R_0)$ and $(\v{c}_1,R_1)$. Here, the parameters for initial and fixed states are chosen to be $x=1/2$, $z=0$ and $\zeta=1/4$.}
\end{figure}
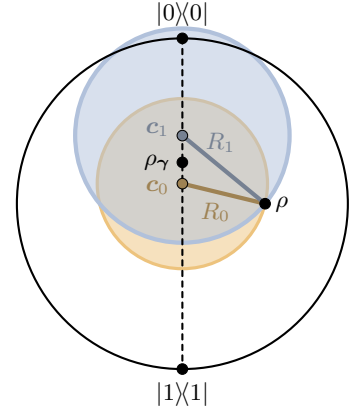

The crucial observation we make here is as follows. Consider the case $z' \geq z$. Divide the extremal path into $n$ parts by choosing states \mbox{$\rho_0,\dots,\rho_n$} along the $(\v{c}_0, R_0)$ circle with $\rho_0=\rho$. Since Eq.~\eqref{eq:r+-} is satisfied, for each $i\in\{0,\dots,n-1\}$ there exists $\mathcal{E}_i$ with $\mathcal{E}_i(\rho_i)= \rho_{i+1}$ and $\mathcal{E}_i(\rho_{\v{\gamma}}) = \rho_{\v{\gamma}}$. Similar considerations hold for $z' < z$ considering the $(\v{c}_1, R_1)$ circle. This suggests that there indeed exists a continuous Markov evolution that evolves the state along the extremal path. 

To construct a time-dependent Lindbladian that evolves the state along the extremal path (say, the one with $z' \geq z$) we fix some arbitrarily small $\Delta>0$ and find the state $\rho_1$ on the extremal path with $z'=z + \Delta$. Using the construction of Ref.~\cite{alberti1980problem} we obtain an explicit form for the quantum channel $\E_0$ mapping $\rho_0$ to $\rho_1$, while preserving $\rho_{\v{\gamma}}$. Next, we define the Lindbladian $\L_0=\E_0-\I$ and evolve the state according to $e^{\L_0}$, obtaining $\tilde{\rho}_1:=e^{\mathcal{L}_0} \rho_0$. We then repeat the same procedure, but instead of $\rho_0$ we start with $\tilde{\rho}_i$ for $i>0$. In this way we construct a whole set of Lindbladians ${\L}_i$. The procedure ends when Eq.~\eqref{eq:r+-} is no longer satisfied for  $z' = z + \Delta$. Due to the extremely complicated form of the Kraus operators describing the channels ${\E}_i$ (and hence ${\L}_i$), instead of their explicit expressions we provide their construction in Appendix~\ref{app:GP_qubit}.

We have thus constructed a quantum Markovian evolution $\prod_i e^{{\L}_i}$ passing through the points $\tilde{\rho}_i$.  Numerical investigations show that this Markovian dynamics evolves $\rho_0=\rho$ approximately along the extremal circle $(\v{c}_0, R_0)$ (or $(\v{c}_1, R_1)$ for $\Delta <0$), with the approximation improving as $\Delta \rightarrow 0$. We illustrate these results for particular choices of initial and fixed states in Fig.~\ref{fig:qubit_path}, and note that this is strong evidence that Result~\ref{result2} holds.

\begin{figure}[t]
	\begin{tikzpicture}[line cap=round,line join=round,x=2cm,y=2cm]
	
		%Bloch spheres
		\draw [line width=0.8pt] (-1.05,0.) circle (2cm);
		\draw [line width=0.8pt] (1.05,0.) circle (2cm);
		
		%z axes of Bloch spheres
		\draw [line width=0.8pt,dash pattern=on 2pt off 2pt] (-1.05,1.)-- (-1.05,-1.);
		\draw [line width=0.8pt,dash pattern=on 2pt off 2pt] (1.05,1.)-- (1.05,-1.);
		\begin{scriptsize}
		
		%Panel a: computational states
		\draw [fill=black] (-1.05,-1.) circle (1.5pt);
		\draw [fill=black] (-1.05,1.) circle (1.5pt);
		
		%Panel b: computational states
		\draw [fill=black] (1.05,-1.) circle (1.5pt);
		\draw [fill=black] (1.05,1.) circle (1.5pt);
		
		%Panel a: extremal circle, radius and centre
		\draw [line width=2pt,color=tikzOrange,fill=tikzOrange,fill opacity=0.5] (-1.05,-3/9) arc (-90:90:5/9);
		\draw [line width=1.6pt,color=tikzOrange2] (-1.05,2/9)-- (-0.5,2/9);
		\draw [fill=tikzOrange2] (-1.05,2/9) circle (1.5pt);
		
		%Panel b: extremal circle, radius and centre 
		\draw [line width=2pt,color=tikzBlue, fill=tikzBlue, fill opacity=0.5] (1.05,-0.1) arc (-90:90:7/15);
		\draw [line width=1.6pt,color=tikzBlue2] (1.05,11/30)-- (1.5,11/30);
		\draw [fill=tikzBlue2] (1.05,11/30) circle (1.5pt);
		
		%Direction arrows
		\draw [->,line width=1pt,color=black] (-3/9,1/24) arc (-45:45:2/9);	
		\draw [<-,line width=1pt,color=black] (1.66,0.2) arc (-45:45:2/9);
		
		%Panel a: actual Markovian path
		\foreach \Point in {(-0.977365,-0.3212),(-0.893765,-0.29867),(-0.836197,-0.275724),(-0.790904,-0.252435),(-0.75326,-0.22894),(-0.721028,-0.205312),(-0.692929,-0.181593),(-0.66815,-0.157811),(-0.646142,-0.133982),(-0.626509,-0.110117),(-0.608962,-0.086225),(-0.593278,-0.0623127),(-0.579283,-0.0383849),(-0.56684,-0.0144454),(-0.55584,0.00950248),(-0.546194,0.033456),(-0.537828,0.0574129),(-0.530685,0.0813709),(-0.524717,0.105328),(-0.519886,0.129283),(-0.516163,0.153233),(-0.513525,0.177177),(-0.511957,0.201113),(-0.511449,0.22504),(-0.511997,0.248955),(-0.513605,0.272857),(-0.516281,0.296742),(-0.520039,0.32061),(-0.524901,0.344457),(-0.530894,0.36828),(-0.538056,0.392076),(-0.546431,0.41584),(-0.556078,0.439568),(-0.567065,0.463254),(-0.579479,0.48689),(-0.593426,0.510468),(-0.609038,0.533976),(-0.626483,0.557399),(-0.645974,0.580719),(-0.667791,0.603909),(-0.692312,0.626929),(-0.720062,0.649725),(-0.751817,0.672205),(-0.788802,0.694213),(-0.833177,0.715443),(-0.88959,0.735149),(-0.976726,0.750134),(-1.02977,0.752856),(-1.04441,0.753063),(-1.04846,0.753079)}{\draw [color=black,fill=black] \Point circle (0.85pt);};
		
		%Panel b: actual Markovian path
		\foreach \Point in {(1.11091,0.823156),(1.177,0.80568),(1.22308,0.787877),(1.25997,0.769593),(1.29101,0.750955),(1.31784,0.732065),(1.34139,0.71299),(1.36228,0.693776),(1.38092,0.674454),(1.39763,0.655047),(1.41262,0.635574),(1.42609,0.616047),(1.43817,0.596475),(1.44896,0.576869),(1.45856,0.557233),(1.46705,0.537575),(1.47447,0.517898),(1.48089,0.498207),(1.48633,0.478505),(1.49084,0.458797),(1.49444,0.439085),(1.49715,0.419371),(1.49898,0.39966),(1.49995,0.379953),(1.50005,0.360254),(1.49929,0.340564),(1.49767,0.320888),(1.49518,0.301229),(1.49181,0.281589),(1.48753,0.261972),(1.48232,0.242383),(1.47616,0.222826),(1.46901,0.203306),(1.46082,0.183831),(1.45154,0.164407),(1.44109,0.145042),(1.4294,0.125749),(1.41637,0.10654),(1.40186,0.0874322),(1.38571,0.0684483),(1.36772,0.0496176),(1.34761,0.0309807),(1.325,0.0125956),(1.29934,-0.00545103),(1.2698,-0.0230206),(1.23497,-0.0398717),(1.1919,-0.0555161),(1.13034,-0.0684892),(1.07362,-0.0723749),(1.05695,-0.0727102)}{\draw [color=black,fill=black] \Point circle (0.85pt);};
		
		%Initial and thermal states
		\draw [fill=black] (1.05,5/6) circle (1.5pt);
		\draw [fill=black] (-1.05,-1/3) circle (1.5pt);
		\draw [fill=black] (-1.05,1/2) circle (1.5pt);
		\draw [fill=black] (1.05,1/4) circle (1.5pt);
		\end{scriptsize}
		
		%Labels
		\node at (-1.05,1.2) {\color{black} $\ketbra{0}{0}$};
		\node at (-1.05,-1.2) {\color{black} $\ketbra{1}{1}$};
		\node at (1.05,1.2) {\color{black} $\ketbra{0}{0}$};
		\node at (1.05,-1.2) {\color{black} $\ketbra{1}{1}$};
		\node at (-1.2,0.5) {\color{black} $\rho_{\v{\gamma}}$};
		\node at (0.9,0.2167) {\color{black} $\rho_{\v{\gamma}}$};
		\node at (-1.2,-0.33333) {\color{black} $\rho$};
		\node at (0.9,0.83333) {\color{black} $\rho$};
		\node at (-1.2,2/9) {\color{tikzOrange2} $\v{c}_0$};
		\node at (0.9,0.4) {\color{tikzBlue2} $\v{c}_1$};	  
		\node at (-0.8,0.3222) {\color{tikzOrange2} $R_0$};
		\node at (1.25,0.47) {\color{tikzBlue2} $R_1$};	 
	\end{tikzpicture}
	\caption{\label{fig:qubit_path} \textbf{Qubit memoryless accessibility region.} Initial state $\rho$ evolved by $\prod_i e^{\L_i}$ (black dots) as described in the main text. The evolved states clearly approach the extremal paths for thermodynamic processes with memory given by circles $(\v{c}_0,R_0)$ and $(\v{c}_1,R_1)$. In both panels $\Delta$ chosen such that the evolution is divided into 100 equal steps (to increase readability only the even steps are plotted). (a) Parameters: $x=0$, $z=-1/3$ and $\zeta=1/2$. (b) Parameters: $x=0$, $z=5/6$, $\zeta=1/4$.}
\end{figure}
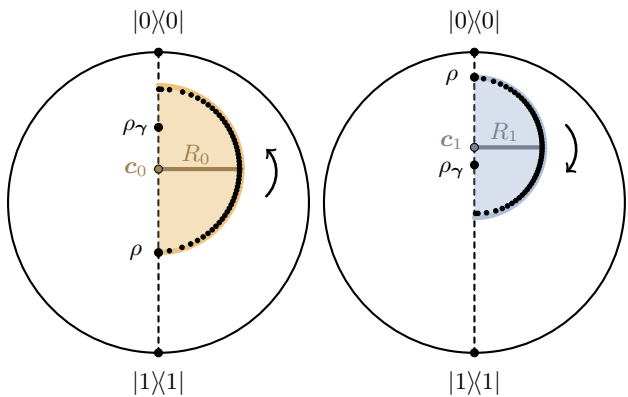

%-------------------------------------------------------------------

\section{Potential for practical advantages}
\label{sec:applications}

While the quantum advantages described in this paper may ultimately find practical applications in information processing technologies, we do not think they can be immediately applied to improve the performance of modern day classical information processors. Quantum advantages proven here arise only once we get to the fundamental limit of encoding classical bits in quantised levels of memory systems. As mentioned in Sec.~\ref{sec:intro:b}, whenever a bit is encoded in macroscopic degrees of freedom (that actually consist of $d\gg 2$ distinct quantum levels) one can employ a subset of them to act as an internal memory and, e.g., unlock the possibility of a classical Markovian bit-swap. Moreover, from a practical perspective, the cost of coherent control over all degrees of freedom of a memory system is massive and so far seems to greatly exceed the benefits we have proven here. The situation in some sense resembles the one with the famous Landauer's principle~\cite{landauer1961irreversibility}: while switching from irreversible to reversible computations one can go below the fundamental $k_BT \log 2$ limit of dissipated heat per use of an AND gate~\cite{bennett1973logical}, the technological cost related to such a change still exceeds possible gains~\cite{koomey2010implications}. Nevertheless, if one day we have as good coherent control over quantum systems as we have incoherent control over classical systems, we will need to use fewer memory states to implement certain processes. 

Despite the above outlook, we still believe that our results may find near-term applications on the side of control theory and quantum thermodynamics. As we discussed, for $d=2$, any classical Markovian master equation with a fixed point $\v{\gamma}$ evolves $\v{p}$ along the path $\v{p}(t)$ that can never go ``on the other side of the fixed point'' (recall Fig.~\ref{fig:cooling}); memory is required for that to happen. On the other hand, the corresponding quantum Markovian master equations access all states achievable under general stochastic maps with a fixed point $\v{\gamma}$. Creation of quantum coherence is crucial since it opens new pathways that ``go around'' the fixed state. What is surprising is that the creation of coherence in a qubit Markovian dissipative process can replace \emph{all} memory effects. Even the ``$\beta$-swap'', the classical process with a thermal fixed point that requires the largest free energy back-flow and achieves the farthest accessible state on the other side of $\v{\gamma}$ in Eq.~\eqref{eq:achievablenomemory}, can be approximated arbitrarily well by a quantum Markovian master equation with a thermal fixed point. 

This observation potentially has significant consequences in the context of cooling, a central problem in quantum sciences. Access to pure, ``cold'' quantum systems is a preliminary requirement for quantum computing~\cite{divincenzo2000physical}. One possible way to achieve it is to employ heat-bath algorithmic cooling (HBAC) protocols~\cite{schulman1999molecular,boykin2002algorithmic, rodriguez2020novel}, initially developed in the context of NMR systems. These are optimised to achieve the largest possible cooling of a target system by means of a sequence of unitary operations and interactions with a heat bath. Recently, optimal HBAC were derived~\cite{rodriguez2017correlation, alhambra2019heat}. For a single qubit, these involve performing a sequence of $X$-pulses and $\beta$-swaps, leading to an exponential convergence to the ground state in the number of cooling rounds~\cite{alhambra2019heat}. What is more, it was also noted that if the interactions with the bath are restricted to be Markovian and destroy superpositions, then one cannot cool the target system below the environment's temperature. 
	
In the scheme of Ref.~\cite{alhambra2019heat} it is the memory effects, encoded in system-bath correlations, that are central in achieving the desired cooling performance. Our Result~\ref{res:1} implies that this control over system-bath correlations can be replaced by the control over the coherent properties of the system. Specifically, a time-dependent Lindbladian with a thermal fixed point is able to generate the required $\beta$-swap. This implies that qubit systems can be cooled all the way to their ground state without exploiting any memory effects. This is a purely quantum phenomenon, exploiting control over coherent resources. The practical impact of these protocols, of course, depends on the possibility of experimentally realizing these exotic dynamics, which is an important question put forward by this investigation. 
	
Moreover, the phenomenon identified in the previous section for qubit systems may be relevant in a much broader setup, as the following general mechanism illustrates. Suppose for simplicity that there are no degeneracies in the Bohr spectrum of the system, i.e., the allowed energy differences, $\{E_i - E_j\}_{i \neq j}$, for the studied system are all distinct. Given the initial classical state of the system, $\rho(0)=\rho_{\v{p}}$, and its evolution described by $\rho(t)$, use the following decomposition:
\begin{equation}
	\rho(t) = \rho_{\v{r}(t)} + C(t),
\end{equation}
where $\v{r}(t)$ is the population in the energy eigenbasis and $C(t)$ denotes the off-diagonal terms (``coherences'') at time $t$. Any classical Markovian evolution with a thermal fixed point requires $\frac{d}{dt} F(\v{r}(t)) \leq 0$ and $C(t) = 0$ at all $t \geq 0$. Any quantum Markovian dynamics requires $\frac{d}{dt} F_Q(\rho(t)) \leq 0$ for all $t\geq0$, where $F_Q$ is the quantum non-equilibrium free energy:
\begin{equation}
	F_Q(\rho) = \tr{\rho H_S} - \beta^{-1} S(\rho), 
\end{equation}  
with $H_S$ denoting the Hamiltonian of the system and \mbox{$S(\rho) = - \tr{\rho \log \rho}$} being the von Neumann entropy. Recall that $F_Q(\rho) $ can be additively decomposed into two non-negative components~\cite{lostaglio2019introductory}:
\begin{equation}
	F_Q(\rho(t)) = F(\v{r}(t)) + \beta^{-1} A(\rho(t)).
\end{equation}

The first term is the classical non-equilibrium free energy and the second one is a quantum component (called ``asymmetry''), which measures the coherent contribution to $F_Q$~\cite{lostaglio2015description}. At $t=0$ we have $\v{r}(0) = \v{p}$ and $C(0) = 0$, which implies $A(\rho(0)) = 0$. Hence, both classical and quantum free energies for the initial state are equal to~$F(\v{p})$. However, a Markovian quantum evolution can store some free energy in coherences at times $t\geq 0$, since only the \emph{sum} of the classical and quantum components of the free energy must monotonically decrease in time. This way, at time $t_*$ when $\v{r}(t_*) = \v{\gamma}$, classically one is stuck in a free energy minimum $F(\v{\gamma})$ and cannot proceed further. But quantum mechanically one can have $C(t_*) \neq 0$ and hence $F_Q(\rho(t_*)) > F(\v{\gamma})$. Markovian quantum dynamics can hence access other states at $t>t_*$ by converting back some of the quantum component of free energy into classical free energy. This allows one to achieve the required backflow into the classical component of the free energy, as illustrated Fig.~\ref{fig:freeenergysketch}. Storing free energy in coherence is of course a non-trivial task (it requires the aid of an external source of coherence and certainly it cannot be done with thermal operations~\cite{lostaglio2015description}), but our qubit construction has shown that it is possible with quantum Markovian master equations.

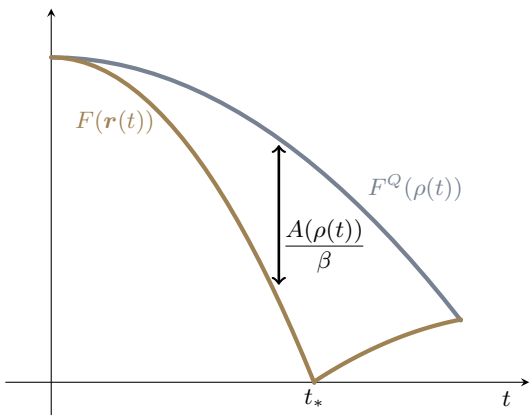
\begin{figure}[t]
	\centering
	\begin{tikzpicture}[line cap=round,line join=round,x=\columnwidth,y=\columnwidth]
	\begin{axis}[
	x=0.7\columnwidth,y=0.5\columnwidth,
	axis lines=middle,
	xmin=-0.1,
	xmax=1.05,
	ymin=-0.1,
	ymax=1.15,
	ticks=none]
	\clip(-0.1,-0.1) rectangle (1.05,1.15);
	
	\draw [-,line width=1.6pt,color=tikzBlue2,domain=0:0.9] plot({\x},{1-\x*\x});
	
	\draw [-,line width=1.6pt,color=tikzOrange2,domain=0:1/1.73] plot({\x},{1-3*\x*\x});
	
	\draw [-,line width=1.6pt,color=tikzOrange2,domain=1/1.73:0.9] plot({\x},{-\x*\x+2.07*\x-0.86});
	
	\draw [<->,line width=1.pt,color=black] (0.5,0.3)-- (0.5,0.73);
	
	\draw[color=tikzBlue2] (0.8,0.6) node {$F^Q(\rho(t))$};
	\draw[color=tikzOrange2] (0.14,0.8) node {$F(\v{r}(t))$};
	\draw[color=black] (0.6,0.42) node {$\dfrac{A(\rho(t))}{\beta}$};
	
	\draw[color=black] (1,-0.05) node {$t$};
	\draw[color=black] (1/1.73,-0.05) node {$t_*$};
	\end{axis}
	\end{tikzpicture}
	\caption{\label{fig:freeenergysketch}\textbf{Free energy stored in coherence.} Under a Markovian master equation with a thermal fixed point, the quantum free energy $F^Q$ is monotonically decreasing in time. Since $F^Q = F +  A/\beta$, part of the classical component $F$ can be stored in the coherent component $A$ at times $t \leq t_*$ and recovered later. At time $t=t_*$ the population is thermal, so the classical free energy is at a minimum, but $A \neq 0$ so $F^Q$ is above the minimum. For $t\geq t_*$ part of this coherent free energy is converted back into classical free energy. Hence, the latter undergoes a backflow which classically would require memory effects.}
\end{figure}

%------------------------------------------------------------------

\section{Outlook}
\label{sec:outlook}

The central task of this paper was to introduce a unifying framework to prove that quantum dynamics offers memory improvements over the classical stochastic evolution in a variety of settings, from the more computational ones to the more physical ones. The unifying notion is that of the underlying Markovian master equation and the advantages one gains with quantum controls as compared to the classical ones. The driving force behind these advantages is the superposition principle, which provides a wider arena for memoryless evolutions to unfold. While Holevo's theorem~\cite{nielsen2010quantum} prevents us from retrieving more than $n$ bits of information from $n$ qubits, we see that in many other respects superpositions can take over the role played by a classical memory. 

This is most clearly captured by the notion of quantum embeddable stochastic processes introduced here: processes which do not require any memory quantum mechanically but they do classically. We found several classes of such processes, but the full characterisation is left as a big open problem for future research. It may be especially hard taking into account that the classical version of the problem is still unsolved for $d>3$, however recent progress on accessibility of quantum channels via the Lindblad semigroup~\cite{shahbeigi2020log} is promising. Moreover, one may still hope for a partial characterisation, e.g., it would be of particular interest to identify the outer limits to the quantum advantage by means of necessary conditions for quantum embeddability. 

We have also proved the quantum advantage in terms of memory and time-step cost of implementing a given stochastic process. By means of computational basis input states and measurements, a quantum Markovian master equation on a $(d+m_Q)$-dimensional quantum system realises a $d\times d$ stochastic process $P$ in $\tau_Q$ time-steps. We compared $(m_Q,\tau_Q)$ with the minimal $(m_C,\tau_C)$ required in order to simulate the same $P$ through any classical Markovian master equation. When $P$ is deterministic, i.e., it is a function over a discrete state space, we saw that typically one has a gap between $(m_Q,\tau_Q)$ and $(m_C,\tau_C)$. 

From a technological perspective, our investigations lead to the question whether these in-principle simulation advantages translate into practical ones. There is a long history, dating back to Landauer, of associating a cost to erasure~\cite{landauer1961irreversibility}, due to the unavoidable dissipation of entropy required for the implementation of such processes. Notably, these costs are common to classical and quantum scenarios. On the other hand, we showed here that the memory and time-step costs of computations under memoryless quantum dynamics are typically lower than the corresponding classical costs. However, practical advantages can be expected only once we overcome technological barriers concerning (a) information density, and (b) coherent control of quantum systems. As we do not expect this to happen within the next decades, we decided not to discuss this point as a potential application.

It would be of interest to extend our framework to sequential scenarios, where one samples $P_{i|j}(t_k)$ over the visible states $i,j$ at discrete times $t_k$, and looks for a predictive dynamical model matching the observed data. We showed that quantum models can reproduce the same data with advantages in terms of the number of hidden states and time-steps. It should be investigated whether such an advantage becomes larger in sequential scenarios. If these translate in practical advantages which can be realised in near-term quantum devices is an intriguing open question of the present investigation.

Finally, our results may also find applicability in the realm of quantum thermodynamics. Our proof that qubit systems can be quantum mechanically cooled below the environmental temperature using neither memory effects nor ancillas (something that is impossible classically) suggests that we should look for realistic Markovian master equations in which this phenomenon can be observed. One could also investigate possible memory advantages in control (and especially cooling) for higher dimensional systems. 

%------------------------------------------------------------------

\section*{Acknowledgements}

The authors would like to thank Stephen Bartlett for hosting them at the initial stage of the project, as well as Iacopo Carusotto and Claudio Fontanari for their kind hospitality in Trento, where part of this project was carried out. KK would also like to thank K.~\.{Z}yczkowski for helpful discussions. KK acknowledges financial support by the Foundation for Polish Science through IRAP project co-financed by EU within Smart Growth Operational Programme (contract no. 2018/MAB/5) and through TEAM-NET project (contract no. POIR.04.04.00-00-17C1/18-00). ML acknowledges financial support from the the European Union's Marie Sk{\l}odowska-Curie individual Fellowships (H2020-MSCA-IF-2017, GA794842), Spanish MINECO (Severo Ochoa SEV-2015-0522 and project QIBEQI FIS2016-80773-P), Fundacio Cellex and Generalitat de Catalunya (CERCA Programme and SGR 875) and ERC grant EQEC No. 682726.

%------------------------------------------------------------------

\appendix

\section{Typical functions}
\label{app:typical}

Consider a function $f: \mathbb{Z}_d \rightarrow \mathbb{Z}_d$ sampled uniformly from the set of all such functions. First, focus on the dimension of the image. The probability that a given \mbox{$a \in \mathbb{Z}_d$} is not in the image of $f$ is $(1-1/d)^d \approx 1/e$ for large $d$. The average dimension of the image is hence a binomial with average $d\left(1-\frac{1}{e}\right)$ and variance $\frac{d}{e}\left(1-\frac{1}{e}\right)$. Hence, for large $d$, the size of the image of $f$ is $O(d)$. Second, focus on the number of fixed points. The probability that a given $a \in \mathbb{Z}_d$ is a fixed point is $1/d$. The number of fixed points is hence a binomial with average $\frac{1}{d} \times d = 1$ and variance $1 \times (1-\frac{1}{d})$. As such, for large $d$ the number of fixed points is $O(1)$. 

%------------------------------------------------------------------

\section{Proof of Theorem~\ref{thm:q_trade-off}}
\label{app:trade_off}

\begin{proof}
	Given a function \mbox{$f: \mathbb{Z}_d \rightarrow \mathbb{Z}_d$} let us denote the size of the image of $f$ by $r=|{\rm img}(f)|$. Next, we denote the elements of ${\rm img}(f)$ by $\{y_k\}_{k=1}^r$, and the remaining elements belonging to \mbox{$\mathbb{Z}_d\setminus {\rm img}(f)$} by $\{y_k\}_{k=r+1}^d$. Moreover, for each $y_k$ with $k\leq r$, i.e., for each of the $r$ elements of the image of $f$, let us denote the corresponding pre-image as follows:
	\begin{align}
		f^{-1}(y_k)=\{x_j^k\}_{j=1}^{d_k}.
	\end{align}
	Note that the sets $\{x_j^k\}_{j=1}^{d_k}$ are disjoint and that their union is the full set $\mathbb{Z}_d$. 
	
	Now, we will construct a permutation function $f_\pi$ and an idempotent function $f_I$, both mapping $\mathbb{Z}_d$ to $\mathbb{Z}_d$ and such that
	\begin{equation}
		\label{eq:f_decomp}
		f=f_I\circ f_\pi.
	\end{equation}
	First, $f_\pi$ is defined by
	\begin{equation}
		\!f_\pi (x_j^k)=\left\{ 
		\begin{array}{ll}
			y_k~& \mathrm{if~} j\!=\!1,\\
			y_{r +\sum_{l=1}^{k-1}(d_l-1)+j-1} ~& \mathrm{otherwise},
		\end{array}
		\right.
	\end{equation}
	where the convention is that $\sum_{l=1}^0 \equiv 0$. Then, introducing \mbox{$n_s:=r+\sum_{l=1}^{s-1}(d_l-1)$}, the idempotent map $f_I$ is given by
	\begin{equation}
		\!f_I (y_k)=\left\{ 
		\begin{array}{ll}
		y_k~& \mathrm{for~} k\!\leq\!r,\\
		y_s~& \mathrm{for~}k\in\{n_s+1,n_s+d_s-1\}.
		\end{array}
		\right.
	\end{equation}
	With the above definitions, it is a straightforward calculation to show that Eq.~\eqref{eq:f_decomp} holds. 
	
	Finally, we need to show that there exist time-independent Lindbladians $\L_\pi$ and $\L_I$ that generate $P_{f_\pi}$ and $P_{f_I}$, i.e., $\{0,1\}$-valued stochastic matrices realising functions $f_\pi$ and $f_I$, respectively. This way, by Definition~\ref{def:quantum_time}, we will prove that any function $f$ can be realised quantumly without the use of memory and in at most 2 time-steps. First, since $P_{f_\pi}$ is a permutation, its generator $\L_\pi$ exists and is simply given by the commutator with the Hamiltonian (see discussion in Sec.~\ref{sec:examples}). Now, in the case of $f_I$, notice that it is a function sending $r$ disjoint sets $Y_k$ of size $d_k$,
	\begin{equation}
		Y_k = y_k\cup \{y_l\}_{l=n_k+1}^{n_k+d_k-1},
	\end{equation} 
	to a single element $y_k$ of the given set $Y_k$. This mapping can be easily realised for $t_f\to\infty$ by a classical generator $L$ (so also by the corresponding quantum Lindbladian) given by
	\begin{equation}
		L_{y_k|y_l}=\left\{
		\begin{array}{ll}
			-1~& \mathrm{for~}k=l\mathrm{~and~}l> r, \\
			1~&  \mathrm{for~}k=f_I(l)\mathrm{~and~}l> r,\\
			0~&\mathrm{otherwise}.
		\end{array}
		\right.
	\end{equation}
\end{proof}

%------------------------------------------------------------------

\section{Extremal path}
\label{app:GP_qubit}

Consider the initial qubit state $\rho$ described by the Bloch vector \mbox{$(x,0,z)$}, and a fixed state $\rho_{\v{\gamma}}$ with the Bloch vector \mbox{$(0,0,\zeta)$}. Here, we will show how to construct quantum channels $\E_0$ and $\E_1$ with a fixed point $\rho_{\v{\gamma}}$, and evolving $\rho$ along the extremal circles $(\v{c}_0,R_0)$ and $(\v{c}_1,R_1)$, as derived in Ref.~\cite{korzekwa2017structure} and described in Sec.~\ref{sec:qubit_ext}. More precisely, for a given $\Delta>0$ we look for $\E_0$ that evolves $\rho$ to $\rho'$ with
\begin{equation}
	z'=z+\Delta, \quad x'=\sqrt{R_0^2-(z' - \zeta(1-R_0))^2}.
\end{equation}
Similarly, for a given $\Delta>0$ we look for $\E_1$ that evolves $\rho$ to $\rho'$ with
\begin{equation}
	z'=z-\Delta, \quad x'=\sqrt{R_1^2-(z' - \zeta(1+R_1))^2}.
\end{equation}
Note that in both cases there is a maximal value of $\Delta$ for $x'$ to stay real, and we assume that $\Delta$ is below that maximal value (otherwise the map we are looking for does not exist). Below, we will explain how to construct Kraus operators $\{A_i,B_i,C_i\}$ for $\E_i$, so that
\begin{equation}
	\E_i(\cdot)=A_i(\cdot)A_i^\dagger+B_i(\cdot)B_i^\dagger+C_i(\cdot)C_i^\dagger.
\end{equation}
The construction is based on the general construction for channels mapping between pairs of qubit states provided by Alberti and Uhlmann in Ref.~\cite{alberti1980problem}. 

The first step is to define the following projectors:
\begin{subequations}
\begin{align}
	\ketbra{\psi_0}{\psi_0}&=\frac{1}{R_0}(\rho-(1-R_0)\rho_{\v{\gamma}}),\\
	\ketbra{\psi_0'}{\psi_0'}&=\frac{1}{R_0}(\rho'-(1-R_0)\rho_{\v{\gamma}}),\\
	\ketbra{\psi_1}{\psi_1}&=-\frac{1}{R_1}(\rho-(1+R_1)\rho_{\v{\gamma}}),\\
	\ketbra{\psi_1'}{\psi_1'}&=-\frac{1}{R_1}(\rho'-(1+R_1)\rho_{\v{\gamma}}).
\end{align}
\end{subequations}
The above four projectors are used to define four unitary matrices,
\begin{align}
	U_i&=\ketbra{0}{\psi_i}+\ketbra{1}{\psi_i^\perp},
\end{align}
with $i\in\{0,1\}$ and analogous primed definition for $U_i'$. These are then employed to define four rotated fixed states $\Gamma_i:=U_i \rho_{\v{\gamma}} U_i^\dagger$ and analogously for $\Gamma_i'$. Let us parametrise these states as follows:
\begin{align}
	\Gamma_i&=\begin{pmatrix}
	a_i&\epsilon_i\sqrt{a_i(1-a_i)}\\\epsilon_i\sqrt{a_i(1-a_i)}&1-a_i
	\end{pmatrix}.
\end{align}
These eight parameters are then used to calculate the following eight new parameters:
\begin{subequations}
	\begin{align}
		\alpha_i =& \sqrt{\frac{a_i (1 - a_i')}{a_i' (1 - a_i)}} \cdot\frac{\epsilon_i \epsilon_i'}{1-\frac{a_i}{a_i'}(1-\epsilon_i^2)}, \\
		\beta_i =& 
		\sqrt{\frac{(a_i' - a_i) (1 - a_i')}{(1 - a_i) a_i'}} \cdot\frac{\epsilon_i'}{1-\frac{a_i}{a_i'}(1-\epsilon_i^2)},\\
		\gamma_i =& \sqrt{\frac{(1-\epsilon_i'^2)-\frac{a_i}{a_i'}(1-\epsilon_i^2)}{1-\frac{a_i}{a_i'}(1-\epsilon_i^2)}} \cdot\sqrt{\frac{1 - a_i'}{1 - a_i}},\\
		\omega_i =& \sqrt{\frac{a_i' - a_i}{1 - a_i}}.			
	\end{align}
\end{subequations}
These, in turn, allow us to introduce the following operators:
\begin{align}
\!	A^\star_i = \begin{pmatrix}
	1&0\\0&\alpha_i
	\end{pmatrix},\quad
	B^\star_i = \begin{pmatrix}
	0&\omega_i\\0&\beta_i
	\end{pmatrix},\quad
	C^\star_i = \begin{pmatrix}
	0&0\\0&\gamma_i
	\end{pmatrix},\!
\end{align}
which after unitary rotations yield the final Kraus operators we are looking for:
\begin{align}
	A_i=U_i'^\dagger A^\star_i U_i,\quad B_i=U_i'^\dagger B^\star_i U_i,\quad C_i=U_i'^\dagger C^\star_i U_i.
\end{align}

%------------------------------------------------------------------

\bibliography{Bibliography}

\end{document}